\documentclass[11pt]{article}
\usepackage{amsxtra}
\usepackage{amssymb, amsmath}
\usepackage{amsthm}
\usepackage {hyperref}

\newtheorem{lemma}{Lemma}[section]
\newtheorem{theorem}{Theorem}[section]
\newtheorem{remark}{Remark}[section]

\usepackage[dvips]{graphicx}

\begin{document}
\begin{center}
\textbf{\LARGE{Sequences of Inequalities Among \\New Divergence Measures} }
\end{center}

\smallskip
\begin{center}
\textbf{\large{Inder Jeet Taneja}}\\
Departamento de Matem\'{a}tica\\
Universidade Federal de Santa Catarina\\
88.040-900 Florian\'{o}polis, SC, Brazil.\\
\textit{e-mail: ijtaneja@gmail.com\\
http://www.mtm.ufsc.br/$\sim $taneja}
\end{center}

\begin{abstract}
There are three classical divergence measures exist in the literature on information theory and statistics. These are namely, Jeffryes-Kullback-Leiber \cite{jef, kul} \textit{J-divergence}. Sibson-Burbea-Rao \cite{bur} \textit{Jensen-Shannon divegernce }and Taneja \cite{tan1} \textit{arithemtic-geometric mean divergence}.  These three measures bear an interesting relationship among each other and are based on logarithmic expressions. The divergence measures like \textit{Hellinger discrimination}, \textit{symmetric }$\chi ^2 - $\textit{divergence}, and \textit{triangular discrimination} are also known in the literature and are not based on logarithmic expressions. Past years Dragomir et al. \cite{dsb}, Kumar and Johnson \cite{kuj} and Jain and Srivastava \cite{jas} studied different kind of divergence measures. In this paper, we have presented some more new divergence measures and obtained inequalities relating these new measures made connections with previous ones. The idea of exponential divergence is also introduced.
\end{abstract}



\section{Introduction}

Let
\[
\Gamma _n = \left\{ {P = (p_1 ,p_2 ,...,p_n )\left| {p_i > 0,\sum\limits_{i
= 1}^n {p_i = 1} } \right.} \right\},
\quad
n \geqslant 2,
\]

\noindent
be the set of all complete finite discrete probability distributions. For
all $P,Q \in \Gamma _n $, below we shall consider two groups of
divergence measures.

\smallskip
\noindent
\textbf{$\bullet$ First Group}
\begin{align}
\label{eq1}
\Delta (P\vert \vert Q) & = \sum\limits_{i = 1}^n {\frac{(p_i - q_i )^2}{p_i +
q_i }} ,\\
\label{eq2}
h(P\vert \vert Q) & = \frac{1}{2}\sum\limits_{i = 1}^n {\left( {\sqrt {p_i } -
\sqrt {q_i } } \right)^2} ,\\
\label{eq3}
\Psi (P\vert \vert Q) & = \sum\limits_{i = 1}^n {\frac{(p_i - q_i )^2(p_i +
q_i )}{p_i q_i }} ,\\
\label{eq4}
K_0 (P\vert \vert Q) & = \sum\limits_{i = 1}^n {\frac{(p_i - q_i )^2}{\sqrt
{p_i q_i } }}
\intertext{and}
\label{eq5}
F(P\vert \vert Q) & = \frac{1}{2}\sum\limits_{i = 1}^n {\frac{(p_i^2 - q_i^2
)^2}{\sqrt {(p_i q_i )^3} }} ,
\end{align}

\noindent
\textbf{$\bullet$ Second Group}
\begin{align}
\label{eq6}
B_1 (P\vert \vert Q) & = \sum\limits_{i = 1}^n {\frac{(p_i - q_i )^4}{\sqrt
{(p_i q_i )^3} }},\\
\label{eq7}
B_2 (P\vert \vert Q) & = \sum\limits_{i = 1}^n {\frac{\left( {\sqrt {p_i } -
\sqrt {q_i } } \right)^4}{p_i + q_i }}  ,\\
\label{eq8}
B_3 (P\vert \vert Q) & = \sum\limits_{i = 1}^n {\frac{\left( {\sqrt {p_i } -
\sqrt {q_i } } \right)^4}{\sqrt {p_i q_i } }} ,\\
\label{eq9}
B_4 (P\vert \vert Q) & = \sum\limits_{i = 1}^n {\frac{\left( {p_i - q_i }
\right)^2\left( {\sqrt {p_i } - \sqrt {q_i } } \right)^2}{\left( {p_i + q_i
} \right)\sqrt {p_i q_i } }} ,\\
\label{eq10}
B_5 (P\vert \vert Q) &  = \sum\limits_{i = 1}^n {\frac{\left( {p_i - q_i }
\right)^2\left( {\sqrt {p_i } - \sqrt {q_i } } \right)^2}{p_i q_i }}
\intertext{and}
\label{eq11}
B_6 (P\vert \vert Q) & = \sum\limits_{i = 1}^n {\frac{\left( {p_i - q_i }
\right)^4}{p_i q_i \left( {p_i + q_i } \right)}} .
\end{align}

We observe that the measures appearing in first group are already known in the literature. The first three measures $\Delta (P\vert \vert Q)$, $h(P\vert \vert Q)$ and $\Psi (P\vert \vert Q)$ are respectively known as \textit{triangular discrimination}, \textit{Hellingar's divergence} \textit{and symmetric chi-square divergence}. The measures $K_0 (P\vert \vert Q)$ and $F(P\vert \vert Q)$ are due to Jain and Srivastava \cite{jas} and Kumar and Johnson \cite{kuj} respectively. The measure $B_1 (P\vert \vert Q)$ appearing in the second group is due to Dragomir et al. \cite{dsb}. Other five measures appearing in the second group are new. The measures (\ref{eq3}), (\ref{eq9}) and (\ref{eq10}) are very much similar to each other and the other eight measures are also similar to each other. The measures (\ref{eq6})-(\ref{eq11}) can be written in terms of the measures (\ref{eq1})-(\ref{eq5}). See the expression (\ref{eq28}).

\subsection{Classical Divergence Measures}

All the above eleven measures are without logarithmic expressions. There are three classical divergence measures known in the literature on information theory and statistics are \textit{J-divergence},  \textit{Jensen-Shannon divergence} and  \textit{Arithmetic-Geometric mean divergence} given respectively as
\begin{align}
\label{eq12}
J(P\vert \vert Q) & = \sum\limits_{i = 1}^n {(p_i - q_i )\ln (\frac{p_i }{q_i
})},\\
\label{eq13}
I(P\vert \vert Q) & = \frac{1}{2}\left[ {\sum\limits_{i = 1}^n {p_i \ln \left(
{\frac{2p_i }{p_i + q_i }} \right) + } \sum\limits_{i = 1}^n {q_i \ln \left(
{\frac{2q_i }{p_i + q_i }} \right)} } \right]
\intertext{and}
\label{eq14}
T(P\vert \vert Q) & = \sum\limits_{i = 1}^n {\left( {\frac{p_i + q_i }{2}}
\right)\ln \left( {\frac{p_i + q_i }{2\sqrt {p_i q_i } }} \right)} .
\end{align}

We have the following inequalities \cite{tan3,tan4} \cite{jas}, \cite{kuj} among the measures (\ref{eq1})-(\ref{eq5}) and (\ref{eq12})-(\ref{eq14}).
\[
\frac{1}{4}\Delta (P\vert \vert Q) \leqslant I(P\vert \vert Q) \leqslant
h(P\vert \vert Q) \leqslant \frac{1}{8}J(P\vert \vert Q) \leqslant T(P\vert
\vert Q) \leqslant
\]
\begin{equation}
\label{eq15}
 \leqslant \frac{1}{8}K_0 (P\vert \vert Q) \leqslant \frac{1}{16}\Psi
(P\vert \vert Q) \leqslant \frac{1}{16}F(P\vert \vert Q).
\end{equation}

Some recent applications of Jensen's difference (ref{eq13}) can be seen in Sachlas and Papaioannou \cite{sap}.

\subsection{Exponential Divergence}

For all $P,\;Q \in \Gamma _n $, let consider the following general measure
\begin{equation}
\label{eq16}
K_t (P\vert \vert Q) = \sum\limits_{i = 1}^n {\frac{(p_i - q_i )^{2(t +
1)}}{\left( {p_i q_i } \right)^{\frac{2t + 1}{2}}}} ,\quad t = 0,1,2,3,...
\end{equation}

When $t = 0$, we have the same measure as given in (\ref{eq4}). When $t = 1$, we have $K_1 (P\vert \vert Q) = B_1 (P\vert \vert Q)$. When $2t + 1 = k$, it reduces to one studied by Jain and Srivastava \cite{jas}. We can easily check that the measures $K_t (P\vert \vert Q)$ are convex in the pair of probability distributions $(P,Q) \in \Gamma _n $, $t = 0,1,2,3,...$

\smallskip
Let us write
\begin{equation}
\label{eq17}
E_K (P\vert \vert Q) = \frac{1}{0!}K_0 (P\vert \vert Q) + \frac{1}{1!}K_1
(P\vert \vert Q) + \frac{1}{2!}K_2 (P\vert \vert Q) + \frac{1}{3!}K_3
(P\vert \vert Q) + ...
\end{equation}

The expression (\ref{eq17}) leads us to following \textit{exponential divergence}
\begin{equation}
\label{eq18}
E_K (P\vert \vert Q) = \sum\limits_{i = 1}^n {\frac{(p_i - q_i )^2}{\sqrt
{p_i q_i } }\exp \left( {\frac{(p_i - q_i )^2}{p_i q_i }} \right)} ,\quad
(P,Q) \in \Gamma _n \times \Gamma _n ,
\end{equation}

The eight measures appearing in the inequalities (\ref{eq15}) admits many nonnegative differences. Here our aim to obtain inequalities relating these measures arising due to nonnegative differences from (\ref{eq15}). Also our aim is to bring inequalities among the six measures $B_1 (P\vert \vert Q)$ to $B_6 (P\vert \vert Q)$ and then again study their nonnegative differences. Aim is also to connect the first four  terms of the series (\ref{eq17}) with the known measures. Frequently, we shall use the following two lemmas.

\begin{lemma} If the function $f:[0,\infty ) \to \mathbb{R}$ is convex and normalized, i.e., $f(1) = 0$, then the \textit{f-divergence}, $C_f (P\vert \vert Q)$ given by
\begin{equation}
\label{eq19}
C_f (P\vert \vert Q) =
\sum\limits_{i = 1}^n {q_i f\left( {\frac{p_i }{q_i }} \right)} ,
\end{equation}

\noindent
is nonnegative and convex in the pair of probability distribution $(P,Q) \in \Gamma _n \times \Gamma _n $.
\end{lemma}

\begin{lemma} Let $f_1 ,f_2 :I \subset \mathbb{R}_ + \to \mathbb{R}$ two generating mappings are normalized, i.e., $f_1 (1) = f_2 (1) = 0$ and satisfy the assumptions:

(i) $f_1 $ and $f_2 $ are twice differentiable on $(a,b)$;

(ii) there exists the real constants $\alpha ,\beta $such that $\alpha <
\beta $ and
\[
\alpha \leqslant \frac{f_1 ^{\prime \prime }(x)}{f_2 ^{\prime \prime }(x)}
\leqslant \beta ,
\quad
f_2 ^{\prime \prime }(x) > 0,
\quad
\forall x \in (a,b),
\]

\noindent then we have the inequalities:
\begin{equation}
\label{eq20}
\alpha \mbox{ }C_{f_2 } (P\vert \vert Q) \leqslant C_{f_1 } (P\vert \vert Q)
\leqslant \beta \mbox{ }C_{f_2 } (P\vert \vert Q).
\end{equation}
\end{lemma}

The first Lemma is due to Csisz\'{a}r \cite{csi} and the second is due to author \cite{tan3}. Some interesting properties of (\ref{eq19}) see Taneja and Kumar \cite{tak}.

\section{Convexity of Difference of Divergences}

The inequalities given in (\ref{eq15}) admit 28 nonnegative differences. Convexity of some of these differences is already studied in Taneja \cite{tan4}. Here we shall study convexity of the differences connected with new measures $K_0 (P\vert \vert Q)$ and $F(P\vert \vert Q)$. We can easily check that in all the cases $f_{( \cdot )} (1) = 0$. According to Lemma 1.1, it is sufficient to show the convexity of the functions $f_{( \cdot )} (x)$, i.e., to show that the second order derivative of $f_{( \cdot )} (x)$, i.e, ${f}''_{( \cdot )} (x)
\geqslant 0$ for all $x > 0$. We shall do each part separately. Throughout, it is understood that $x > 0$.

\bigskip
\noindent
\textbf{(i) For }$D_{K_0 T} (P\vert \vert Q)$\textbf{:} We can write
\[
 D_{K_0 T} (P\vert \vert Q)  = \sum\limits_{i = 1}^n {q_i f_{K_0 T} \left(
{\frac{p_i }{q_i }} \right)} ,
\]
\noindent where
\[
f_{K_0 T} (x) = \frac{1}{8}f_{K_0 } (x) - f_T (x)
 = \frac{1}{8}\frac{\left( {x - 1} \right)^2}{\sqrt x } - \frac{x + 1}{2}\ln
\left( {\frac{x + 1}{2\sqrt x }} \right).
\]
\noindent This gives
\[
 {f}''_{K_0 T} (x) = \frac{\left( {3x + 4\sqrt x + 3} \right)\left( {\sqrt x
- 1} \right)^4}{32x^{5 / 2}(x + 1)} \geqslant 0.
\]

\smallskip
\noindent
\textbf{(ii) For }$D_{K_0 J} (P\vert \vert Q)$\textbf{:} We can write
\[
D_{K_0 J} (P\vert \vert Q) = \sum\limits_{i = 1}^n {q_i f_{K_0 J} \left(
{\frac{p_i }{q_i }} \right)} ,
\]
\noindent where
\[
f_{K_0 J} (x) = \frac{1}{8}f_{K_0 } (x) - \frac{1}{8}f_J (x)
 = \frac{1}{8}\frac{\left( {x - 1} \right)^2}{\sqrt x } - \frac{1}{8}\left(
{x - 1} \right)\ln x.
\]
\noindent This gives
\[
{f}''_{K_0 J} (x) = \frac{\left( {3x + 2\sqrt x + 3} \right)\left( {\sqrt x
- 1} \right)^2}{32x^{5 / 2}} \geqslant 0.
\]

\smallskip
\noindent
\textbf{(iii) For }$D_{K_0 h} (P\vert \vert Q)$\textbf{:} We can write
\[
D_{K_0 h} (P\vert \vert Q) = \sum\limits_{i = 1}^n {q_i f_{K_0 h} \left(
{\frac{p_i }{q_i }} \right)} ,
\]
\noindent where
\[
f_{K_0 h} (x) = \frac{1}{8}f_{K_0 } (x) - f_h (x)
 = \frac{1}{8}\frac{\left( {x - 1} \right)^2}{\sqrt x } - \frac{1}{2}\left(
{\sqrt x - 1} \right)^2 = \frac{\left( {\sqrt x - 1} \right)^4}{8\sqrt x }.
\]
\noindent This gives
\[
{f}''_{K_0 h} (x) = \frac{3\left( {x - 1} \right)^2}{32x^{5 / 2}} \geqslant
0.
\]

\smallskip
\noindent
\textbf{(iv) For }$D_{K_0 I} (P\vert \vert Q)$\textbf{:} We can write
\[
D_{K_0 I} (P\vert \vert Q) = \sum\limits_{i = 1}^n {q_i f_{K_0 I} \left(
{\frac{p_i }{q_i }} \right)} ,
\]
\noindent where
\[
f_{K_0 I} (x) = \frac{1}{8}f_{K_0 } (x) - f_I (x)
 = \frac{1}{8}\frac{\left( {x - 1} \right)^2}{\sqrt x } - \frac{x}{2}\ln x +
\frac{x + 1}{2}\ln \left( {\frac{x + 1}{2}} \right).
\]
\noindent This gives
\[
{f}''_{K_0 I} (x) = \frac{\left( {3x^2 + 6x^{3 / 2} + 14x + 6\sqrt x + 3}
\right)\left( {\sqrt x - 1} \right)^2}{32x^{5 / 2}(x + 1)} \geqslant 0.
\]

\smallskip
\noindent
\textbf{(v) For }$D_{K_0 \Delta } (P\vert \vert Q)$\textbf{:} We can write
\[
D_{K_0 \Delta } (P\vert \vert Q) = \sum\limits_{i = 1}^n {q_i f_{K_0 \Delta
} \left( {\frac{p_i }{q_i }} \right)} ,
\]
\noindent where
\[
f_{K_0 \Delta } (x) = \frac{1}{8}f_{K_0 } (x) - \frac{1}{4}f_\Delta (x)
\]
\[
 = \frac{1}{8}\frac{\left( {x - 1} \right)^2}{\sqrt x } -
\frac{1}{4}\frac{\left( {x - 1} \right)^2}{x + 1} = \frac{\left( {x - 1}
\right)^2\left( {\sqrt x - 1} \right)^2}{8\sqrt x \left( {x + 1} \right)}.
\]
\noindent This gives
\[
{f}''_{K_0 \Delta } (x) = \frac{\left( {\sqrt x - 1} \right)^2\left(
{\begin{array}{l}
 3x^4 + 6x^{7 / 2} + 20x^3 + 34x^{5 / 2} + \\
 + 66x^2 + 34x^{3 / 2} + 20x + 6\sqrt x + 3 \\
 \end{array}} \right)}{32x^{5 / 2}(x + 1)^3} \geqslant 0.
\]

\smallskip
\noindent
\textbf{(vi) For }$D_{\Psi K_0 } (P\vert \vert Q)$\textbf{:} We can write
\[
D_{\Psi K_0 } (P\vert \vert Q) = \sum\limits_{i = 1}^n {q_i f_{\Psi K_0 }
\left( {\frac{p_i }{q_i }} \right)} ,
\]
\noindent where
\[
f_{\Psi K_0 } (x) = \frac{1}{16}f_\Psi (x) - \frac{1}{8}f_{K_0 } (x)
\]
\[
 = \frac{1}{16}\frac{\left( {x - 1} \right)^2\left( {x + 1} \right)}{x} -
\frac{1}{8}\frac{\left( {x - 1} \right)}{\sqrt x } = \frac{\left( {x - 1}
\right)^2\left( {\sqrt x - 1} \right)^2}{16x}.
\]
\noindent This gives
\[
{f}''_{\Psi K_0 } (x) = \frac{\left( {\sqrt x - 1} \right)^2\left( {4x^2 +
5x^{3 / 2} + 6x + 5\sqrt x + 4} \right)}{32x^3} \geqslant 0.
\]

\smallskip
\noindent
\textbf{(vii) For }$D_{F\Psi } (P\vert \vert Q)$\textbf{:} We can write
\[
D_{F\Psi } (P\vert \vert Q) = \sum\limits_{i = 1}^n {q_i f_{F\Psi } \left(
{\frac{p_i }{q_i }} \right)} ,
\]
\noindent where
\begin{align}
f_{F\Psi } (x) & = \frac{1}{16}f_F (x) - \frac{1}{16}f_\Psi (x)\notag\\
 & = \frac{1}{32}\frac{\left( {x^2 - 1} \right)^2}{x^{3 / 2}} -
\frac{1}{16}\frac{\left( {x - 1} \right)^2\left( {x + 1} \right)}{x} \notag\\
& =\frac{\left( {x + 1} \right)\left( {\sqrt x + 1} \right)^2\left( {\sqrt x -
1} \right)^4}{32x^{3 / 2}}.\notag
\end{align}
\noindent This gives
\[
{f}''_{F\Psi } (x) = \frac{\left( {\sqrt x - 1} \right)^2\left(
{\begin{array}{l}
 15x^3 + 14x^{5 / 2} + 13x^2 + \\
 + 12x^{3 / 2} + 13x + 14\sqrt x + 15 \\
 \end{array}} \right)}{128x^{7 / 2}} \geqslant 0.
\]

\smallskip
\noindent
\textbf{(viii) For }$D_{FK_0 } (P\vert \vert Q)$\textbf{:} We can write
\[
D_{FK_0 } (P\vert \vert Q) = \sum\limits_{i = 1}^n {q_i f_{FK_0 } \left(
{\frac{p_i }{q_i }} \right)} ,
\]
\noindent where
\begin{align}
f_{FK_0 } (x) & = \frac{1}{16}f_F (x) - \frac{1}{8}f_{K_0 } (x)\notag\\
&  = \frac{1}{32}\frac{\left( {x^2 - 1} \right)^2}{x^{3 / 2}} -
\frac{1}{8}\frac{\left( {x - 1} \right)^2}{\sqrt x } = \frac{\left( {x - 1}
\right)^4}{32x^{3 / 2}}.\notag
\end{align}
\noindent This gives
\[
{f}''_{FK_0 } (x) = \frac{3\left( {5x^2 + 6x + 5} \right)\left( {x - 1}
\right)^2}{128x^{7 / 2}} \geqslant 0.
\]

\smallskip
\noindent
\textbf{(ix) For }$D_{FT} (P\vert \vert Q)$\textbf{:} We can write
\[
D_{FT} (P\vert \vert Q) = \sum\limits_{i = 1}^n {q_i f_{FT} \left(
{\frac{p_i }{q_i }} \right)} ,
\]
\noindent where
\[
f_{FT} (x) = \frac{1}{16}f_F (x) - f_T (x)
 = \frac{1}{32}\frac{\left( {x^2 - 1} \right)^2}{x^{3 / 2}} - \left(
{\frac{x + 1}{2}} \right)\ln \left( {\frac{x + 1}{2\sqrt x }} \right).
\]
\noindent This gives
\[
{f}''_{FT} (x) = \frac{\left( {\sqrt x - 1} \right)^2\left(
{\begin{array}{l}
 15x^4 + 30x^{7 / 2} + 60x^3 + 58x^{5 / 2} + \\
 + 58x^2 + 58x^{3 / 2} + 60x + 30\sqrt x + 15 \\
 \end{array}} \right)}{128x^{7 / 2}(x + 1)} \geqslant 0.
\]

\smallskip
\noindent
\textbf{(x) For }$D_{FJ} (P\vert \vert Q)$\textbf{:} We can write
\[
D_{FJ} (P\vert \vert Q) = \sum\limits_{i = 1}^n {q_i f_{FJ} \left(
{\frac{p_i }{q_i }} \right)} ,
\]
\noindent where
\[
f_{FJ} (x) = \frac{1}{16}f_F (x) - \frac{1}{8}f_J (x)
 = \frac{1}{32}\frac{\left( {x^2 - 1} \right)^2}{x^{3 / 2}} -
\frac{1}{8}\left( {x - 1} \right)\ln x.
\]
\noindent This gives
\[
{f}''_{FJ} (x) =  \frac{\left( {\sqrt x - 1} \right)^2\left(
{\begin{array}{l}
 15x^3 + 30x^{5 / 2} + 45x^2 + \\
 + 44x^{3 / 2} + 45x + 30\sqrt x + 15 \\
 \end{array}} \right)}{128x^{7 / 2}} \geqslant 0.
\]

\smallskip
\noindent
\textbf{(xi) For }$D_{Fh} (P\vert \vert Q)$\textbf{:} We can write
\[
D_{Fh} (P\vert \vert Q) = \sum\limits_{i = 1}^n {q_i f_{Fh} \left(
{\frac{p_i }{q_i }} \right)} ,
\]
\noindent where
\begin{align}
f_{Fh} (x) & = \frac{1}{16}f_F (x) - f_h (x) = \frac{1}{32}\frac{\left( {x^2 -
1} \right)^2}{x^{3 / 2}} - \frac{1}{2}\left( {\sqrt x - 1} \right)^2\notag\\
& = \frac{\left( {x^2 + 4x^{3 / 2} + 10x + 4\sqrt x + 1} \right)\left( {\sqrt
x - 1} \right)^4}{32x^{3 / 2}}.\notag
\end{align}
\noindent This gives
\[
{f}''_{Fh} (x) = \frac{15\left( {x^2 - 1} \right)^2}{128x^{7 / 2}} \geqslant
0.
\]

\smallskip
\noindent
\textbf{(xii) For }$D_{FI} (P\vert \vert Q)$\textbf{:} We can write
\[
D_{FI} (P\vert \vert Q) = \sum\limits_{i = 1}^n {q_i f_{FI} \left(
{\frac{p_i }{q_i }} \right)} ,
\]
\noindent where
\begin{align}
f_{FI} (x) & = \frac{1}{16}f_F (x) - f_I (x) \notag\\
& = \frac{1}{32}\frac{\left( {x^2 -
1} \right)^2}{x^{3 / 2}} - \frac{1}{2}x\ln x + \left( {\frac{x + 1}{2}}
\right)\ln \left( {\frac{x + 1}{2}} \right).\notag
\end{align}
\noindent This gives
\[
{f}''_{FI} (x) = \frac{\left( {\sqrt x - 1} \right)^2\left(
{\begin{array}{l}
 15x^4 + 30x^{7 / 2} + 60x^3 + 90x^{5 / 2} + \\
 + 122x^2 + 90x^{3 / 2} + 60x + 30\sqrt x + 15 \\
 \end{array}} \right)}{128x^{7 / 2}(x + 1)} \geqslant 0.
\]

\smallskip
\noindent
\textbf{(xiii) For }$D_{F\Delta } (P\vert \vert Q)$\textbf{:} We can write
\[
D_{F\Delta } (P\vert \vert Q) = \sum\limits_{i = 1}^n {q_i f_{F\Delta }
\left( {\frac{p_i }{q_i }} \right)} ,
\]
\noindent where
\begin{align}
f_{F\Delta } (x) & = \frac{1}{16}f_F (x) - \frac{1}{4}f_\Delta (x) =
\frac{1}{32}\frac{\left( {x^2 - 1} \right)^2}{x^{3 / 2}} -
\frac{1}{4}\frac{\left( {x - 1} \right)^2}{x + 1}\notag\\
& = \frac{\left( {x - 1} \right)^2\left( {\sqrt x - 1} \right)^2\left( {x^2 +
2x^{3 / 2} + 6x + 2\sqrt x + 1} \right)}{32x^{3 / 2}(x + 1)}.\notag
\end{align}
\noindent This gives
\[
{f}''_{F\Delta } (x) = \frac{\left( {\sqrt x - 1} \right)^2 \left( {\begin{array}{l}
 15x^6 + 30x^{11 / 6} + 90x^5 + 150x^{9 / 2} + \\
 + 257x^4 + 364x^{7 / 2} + 492x^3 + 364x^{5 / 2} + \\
 + 257x^2 + 150x^{3 / 2} + 90x + 30\sqrt x + 15 \\
 \end{array}} \right)}{128x^{7 / 2}(x
+ 1)^3} \geqslant 0.
\]

\begin{remark} In view of above expressions we can relate the measures of group 2 in terms of measures of group 1 as $B_1 (P\vert \vert Q) = 32D_{FK_0 } (P\vert \vert Q)$, $B_3 (P\vert \vert Q) = 8D_{K_0 h} (P\vert \vert Q)$, $B_4 (P\vert \vert Q) = 8D_{K_0 \Delta } (P\vert \vert Q)$, and $B_5 (P\vert \vert Q) = 16D_{\Psi K_0 } (P\vert \vert Q)$. The measures $B_2 (P\vert
\vert Q)$ and $B_6 (P\vert \vert Q)$ can also be written as $B_2 (P\vert \vert Q) = 4D_{h\Delta } (P\vert \vert Q)$ and $B_6 (P\vert \vert Q) = 16D_{\Psi \Delta } (P\vert \vert Q)$.
\end{remark}

\section{Sequences of Inequalities}

The expression (\ref{eq15}) admits 28 nonnegative differences. Some of these differences are already studied in Taneja \cite{tan4}. Here we shall consider only those connected with the measures $K_0 (P\vert \vert Q)$, $\Psi (P\vert \vert Q)$ and $F(P\vert \vert Q)$. Based on these differences the following theorem hold:

\begin{theorem} The following sequences of inequalities hold:
\begin{align}
& \left( {\begin{array}{l}
 \frac{1}{3} D_{T\Delta } \\
 2D_{hI} \\
 \end{array}} \right) \leqslant  \left( {\begin{array}{l}
 \frac{1}{3}D_{K_0 \Delta } \\
 2D_{Jh} \\
 \end{array}} \right) \leqslant
\frac{1}{2}D_{K_0 I} \leqslant \frac{2}{3}D_{K_0 h} \leqslant \left(
{\begin{array}{l}
 D_{K_0 J} \\
 \frac{1}{6}D_{\Psi \Delta } \\
 \end{array}} \right) \leqslant \frac{1}{5}D_{\Psi I} \leqslant  \frac{1}{4}D_{\Psi J}\notag\\
\label{eq21}
& \leqslant \frac{1}{3}\left( {\begin{array}{l}
 D_{\Psi T} \\
 D_{\Psi K_0 } \leqslant \frac{1}{9}D_{F\Delta } \\
 \end{array}} \right) \leqslant \frac{1}{8}D_{FI} \leqslant
\frac{2}{15}D_{Fh} \leqslant \frac{1}{7}D_{FJ} \leqslant \frac{1}{6}\left(
{\begin{array}{l}
 D_{FT} \\
 D_{FK_0 } \\
 \end{array}} \right) \leqslant \frac{1}{3}D_{F\Psi },
\end{align}

\noindent where, for example, $D_{T\Delta } = T - \frac{1}{4}\Delta $, $D_{K_0 \Delta
} = \frac{1}{8}K_0 - \frac{1}{4}\Delta $, etc.
\end{theorem}

\begin{proof} All the measures appearing in the inequalities (\ref{eq21}) can be written as (\ref{eq19}), where we can easily check that all these differences are convex functions in the pair of probability distribution $(P,Q) \in \Gamma _n \times \Gamma _n $. We shall make use of the Lemma 1.2 and shall do each part separately.

\smallskip
\noindent
\textbf{(i) For } $D_{T\Delta } (P\vert \vert Q) \leqslant D_{K_0 \Delta } (P\vert \vert Q)$\textbf{: } For all $x > 0,\;x \ne 1$, let us consider the function
\begin{equation}
\label{eq22}
g_{T\Delta \_K_0 \Delta } (x) = \frac{{f}''_{T\Delta } (x)}{{f}''_{K_0
\Delta } (x)} = \frac{(8x^2 + 32x + 8)\sqrt x \left( {\sqrt x + 1}
\right)^2}{\left( {\begin{array}{l}
 3x^4 + 6x^{7 / 2} + 20x^3 + 34x^{5 / 2} + \\
 + 66x^2 + 34x^{3 / 2} + 20x + 6\sqrt x + 3 \\
 \end{array}} \right)}.
\end{equation}

Calculating the first order derivative of the function $g_{T\Delta \_K_0 \Delta } (x)$ with respect to $x$, $x > 0$, one gets
\begin{align}
{g}'_{T\Delta \_K_0 \Delta } (x) & = - \frac{4\left( {\sqrt x + 1}
\right) \left( {x - 1} \right)}{\sqrt x \left( {\begin{array}{l}
 3x^4 + 6x^{7 / 2} + 20x^3 + 34x^{5 / 2} + \\
 + 66x^2 + 34x^{3 / 2} + 20x + 6\sqrt x + 3 \\
 \end{array}} \right)^2}\times \notag\\
\label{eq23}
& \hspace{20pt} \times \left( {\begin{array}{l}
 3x^5 + 12x^{9 / 2} + 37x^4 + 88x^{7 / 2} + \\
 + 56x^3 + 88x^{5 / 2} + 56x^2 + \\
 + 88x^{3 / 2} + 37x + 12\sqrt x + 3 \\
 \end{array}} \right)
\begin{cases}
 { > 0,} & {0 < x < 1} \\
 { < 0,} & {x > 1} \\
\end{cases}.
\end{align}

In view of (\ref{eq23}) we conclude that the function $g_{T\Delta \_K_0 \Delta } (x)$ is increasing in $x \in (0,1)$ and decreasing in $x \in (1,\infty )$. Also we have
\begin{equation}
\label{eq24}
\beta _{T\Delta \_K_0 \Delta } = \mathop {\sup }\limits_{x \in (0,\infty )}
g_{T\Delta \_K_0 \Delta } (x) = \mathop {\lim }\limits_{x \to 1} g_{T\Delta
\_K_0 \Delta } (x) = 1.
\end{equation}

By the application of (\ref{eq20}) with (\ref{eq24}) we get the required result.

\bigskip
\noindent
\textit{From the above proof we observe that it sufficient to get the expressions similar to (\ref{eq23}) and calculate the value of }$\beta _{ }$ \textit{as given in (\ref{eq24}). For the other parts below we shall avoid all these details. We shall just write the expressions similar to (\ref{eq22}), (\ref{eq23}) and (\ref{eq24}). Then applying the Lemma 1.2, we get the required result. Throughout,  it is understood that }$x > 0,\;x \ne 1.$

\bigskip
\noindent
\textbf{(ii) For }$D_{hI} (P\vert \vert Q) \leqslant \frac{1}{6}D_{K_0 \Delta } (P\vert \vert Q)$\textbf{: }We have
\[
g_{hI\_K_0 \Delta } (x) = \frac{8x\left( {x + 1} \right)^2}{\left(
{\begin{array}{l}
 3x^4 + 6x^{7 / 2} + 20x^3 + 34x^{5 / 2} + \\
 + 66x^2 + 34x^{3 / 2} + 20x + 6\sqrt x + 3 \\
 \end{array}} \right)},
 \]
 \begin{align}
{g}'_{hI\_K_0 \Delta } (x) & = - \frac{8(x - 1)(x + 1)}{\sqrt x \left(
{\begin{array}{l}
 3x^4 + 6x^{7 / 2} + 20x^3 + 34x^{5 / 2} + \\
 + 66x^2 + 34x^{3 / 2} + 20x + 6\sqrt x + 3 \\
 \end{array}} \right)^2}\times\notag\\
& \hspace{20pt} \times \left( {\begin{array}{l}
 3x^4 + 3x^{7 / 2} + 5x^3 + 7x\left( {x - 1} \right)^2 + \\
 + x^{5 / 2} + x^{3 / 2} + 5x + 3\sqrt x + 3 \\
 \end{array}} \right)
\begin{cases}
 { > 0,} & {0 < x < 1} \\
 { < 0,} & {x > 1} \\
\end{cases}\notag
\end{align}
\noindent and
\[
\beta _{hI\_K_0 \Delta }  = \mathop {\sup }\limits_{x \in (0,\infty )}
g_{hI\_K_0 \Delta } (x) = \mathop {\lim }\limits_{x \to 1} g_{hI\_K_0 \Delta
} (x) = \frac{1}{6}.
\]

\smallskip
\noindent
\textbf{(iii) For }$D_{K_0 \Delta } (P\vert \vert Q) \leqslant \frac{3}{2}D_{K_0 I} (P\vert \vert Q)$\textbf{:} We have
\[
g_{K_0 \Delta \_K_0 I} (x)  = \frac{{f}''_{K_0 \Delta } (x)}{{f}''_{K_0 I}
(x)}  = \frac{\left( {\begin{array}{l}
3 x^4 + 6x^{7 / 2} + 20x^3 + 34x^{5 / 2} + \\
 + \,66x^2 + 34x^{3 / 2} + 20x + 6\sqrt x + 3 \\
 \end{array}} \right)}{(x + 1)^2\left( {3x^2 + 6x^{3 / 2} + 14x + 6\sqrt x +
3} \right)},
\]
\begin{align}
{g}'_{K_0 \Delta \_K_0 I} (x) & =
 - \frac{8\left( {x - 1} \right)\sqrt x \left( {3x + 8\sqrt x + 3}
\right)}{\left( {x + 1} \right)^3\left( {3x^2 + 6x^{3 / 2} + 14x + 6\sqrt x
+ 3} \right)^2}\times\notag\\
& \hspace{20pt} \times \left( {3x^2 + 4x^{3 / 2} + 10x + 4\sqrt x + 3} \right)
\begin{cases}
 { > 0,} & {0 < x < 1} \\
 { < 0,} & {x > 1} \\
\end{cases}\notag
\end{align}
\noindent and
\[
\beta _{K_0 \Delta \_K_0 I}  = \mathop {\sup }\limits_{x \in (0,\infty )}
g_{K_0 \Delta \_K_0 I} (x) = \mathop {\lim }\limits_{x \to 1} g_{K_0 \Delta
\_K_0 I} (x) = \frac{3}{2}.
\]

\smallskip
\noindent
\textbf{(iv) For }$D_{Jh} (P\vert \vert Q) \leqslant \frac{1}{4}D_{K_0 I} (P\vert \vert Q)$\textbf{: }We have
\[
g_{Jh\_K_0 I} (x)  = \frac{{f}''_{Jh} (x)}{{f}''_{K_0 I} (x)} = \frac{4\sqrt
x \left( {x + 1} \right)}{3x^2 + 6x^{3 / 2} + 14x + 6\sqrt x + 3},
\]
\[
{g}'_{Jh\_K_0 I} (x)  = - \frac{2(x - 1)\left[ {2(x^2 + 1) + (x - 1)^2} \right]}{\sqrt x \left( {3x^2 + 6x^{3 / 2}
+ 14x + 6\sqrt x + 3} \right)^2}
\begin{cases}
 { > 0,} & {0 < x < 1} \\
 { < 0,} & {x > 1} \\
\end{cases}.
\]

\noindent and
\[
\beta _{Jh \_K_0 I} = \mathop {\sup }\limits_{x \in (0,\infty )}
g_{Jh \_K_0 I} (x) = \mathop {\lim }\limits_{x \to 1} g_{Jh \_K_0 I} (x) = \frac{1}{4}.\notag
\]

\smallskip
\noindent
\textbf{(v) For }$D_{K_0 I} (P\vert \vert Q) \leqslant \frac{4}{3}D_{K_0 h} (P\vert \vert Q)$\textbf{: }We have
\[
g_{K_0 I\_K_0 h} (x)  = \frac{{f}''_{K_0 I} (x)}{{f}''_{K_0 h} (x)} =
\frac{3x^2 + 6x^{3 / 2} + 14x + 6\sqrt x + 3}{3 \left( {x + 1} \right)\left(
{\sqrt x + 1} \right)^2},
\]
\[
{g}'_{K_0 I\_K_0 h} (x)  =
 - \frac{8\left( {\sqrt x - 1} \right)\left( {x + \sqrt x + 1}
\right)}{3\left( {\sqrt x + 1} \right)^3\left( {x + 1} \right)^2}
\begin{cases}
 { > 0,} & {0 < x < 1} \\
 { < 0,} & {x > 1} \\
\end{cases}
\]
\noindent and
\[
\beta _{K_0 I\_K_0 h} = \mathop {\sup }\limits_{x \in (0,\infty )} g_{K_0
I\_K_0 h} (x) = \mathop {\lim }\limits_{x \to 1} g_{K_0 I\_K_0 h} (x) =
\frac{4}{3}.\notag
\]

\smallskip
\noindent
\textbf{(vi) For }$D_{K_0 h} (P\vert \vert Q) \leqslant \frac{3}{2}D_{K_0 J} (P\vert \vert Q)$\textbf{: }We have
\[
g_{K_0 h\_K_0 J} (x)  = \frac{{f}''_{K_0 h} (x)}{{f}''_{K_0 J} (x)} =
\frac{3\left( {\sqrt x + 1} \right)^2}{3x + 2\sqrt x + 3},
\]
\[
{g}'_{K_0 h\_K_0 J} (x)  =  - \frac{6\left( {x - 1} \right)}{\sqrt x \left( {3x + 2\sqrt x + 3}
\right)^2}
\begin{cases}
 { > 0,} & {0 < x < 1} \\
 { < 0,} & {x > 1} \\
\end{cases}
\]
\noindent and
\[
\beta _{K_0 h\_K_0 J}  = \mathop {\sup }\limits_{x \in (0,\infty )} g_{K_0
h\_K_0 J} (x) = \mathop {\lim }\limits_{x \to 1} g_{K_0 h\_K_0 J} (x) =
\frac{3}{2}.
\]

\smallskip
\noindent
\textbf{(vii) For }$D_{K_0 h} (P\vert \vert Q) \leqslant \frac{1}{4}D_{\Psi \Delta } (P\vert \vert Q)$\textbf{: }We have
\[
g_{K_0 h\_\Psi \Delta } (x) = \frac{{f}''_{K_0 h} (x)}{{f}''_{\Psi \Delta }
(x)} = \frac{3\sqrt x \left( {x + 1} \right)^3}{4\left( {x^4 + 5x^3 + 12x^2
+ 5x + 1} \right)},
\]
\[
{g}'_{K_0 h\_\Psi \Delta } (x)  =  - \frac{3(x - 1)^3(x + 1)^2(x^2 + 5x + 1)}{8\sqrt x \left( {x^4 + 5x^3 +
12x^2 + 5x + 1} \right)^2}
\begin{cases}
 { > 0,} & {0 < x < 1} \\
 { < 0,} & {x > 1} \\
\end{cases}
\]
\noindent and
\[
\beta _{K_0 h\_\Psi \Delta }  = \mathop {\sup }\limits_{x \in (0,\infty )}
g_{K_0 h\_\Psi \Delta } (x) = \mathop {\lim }\limits_{x \to 1} g_{K_0
h\_\Psi \Delta } (x) = \frac{1}{4}.
\]

\smallskip
\noindent
\textbf{(viii) For }$D_{K_0 J} (P\vert \vert Q) \leqslant \frac{1}{5}D_{\Psi I} (P\vert \vert Q)$\textbf{: }We have
\[
g_{K_0 J\_\Psi I} (x) = \frac{{f}''_{K_0 J} (x)}{{f}''_{\Psi I} (x)} =
\frac{\sqrt x \left( {3x + 2\sqrt x + 3} \right)\left( {x + 1}
\right)}{4\left( {x^2 + 3x + 1} \right)\left( {\sqrt x + 1} \right)^2},
\]
\begin{align}
{g}'_{K_0 J\_\Psi I} (x) &  = - \frac{\left( {\sqrt x - 1} \right)\left( {3x + \sqrt x + 3}
\right)}{8\sqrt x \left( {x^2 + 3x + 1} \right)^2\left( {\sqrt
x + 1} \right)^3}\times\notag\\
& \hspace{20pt}\times \left( {\begin{array}{l}
 x^3 + x^{5 / 2} + 2x\left( {x + 1} \right) + \\
 + x\left( {\sqrt x - 1} \right)^2 + \sqrt x + 1 \\
 \end{array}} \right)
\begin{cases}
 { > 0,} & {0 < x < 1} \\
 { < 0,} & {x > 1} \\
\end{cases} \notag
\end{align}
\noindent and
\[
\beta _{K_0 J\_\Psi I}  = \mathop {\sup }\limits_{x \in (0,\infty )} g_{K_0
J\_\Psi I} (x) = \mathop {\lim }\limits_{x \to 1} g_{K_0 J\_\Psi I} (x) =
\frac{1}{5}.
\]

\smallskip
\noindent
\textbf{(ix) For }$D_{\Psi J} (P\vert \vert Q) \leqslant \frac{4}{3}D_{\Psi K_0 } (P\vert \vert Q)$\textbf{: }We have
\[
g_{\Psi J\_\Psi K_0 } (x) = \frac{{f}''_{\Psi J} (x)}{{f}''_{\Psi K_0 } (x)}
= \frac{4\left( {\sqrt x + 1} \right)^2\left( {x + 1} \right)}{4x^2 + 5x^{3
/ 2} + 6x + 5\sqrt x + 4},
\]
\[
{g}'_{\Psi J\_\Psi K_0 } (x) = - \frac{2\left( {x - 1} \right)\left( {3x^2 +
4x^{3 / 2} + 10x + 4\sqrt x + 3} \right)}{\sqrt x \left( {4x^2 + 5x^{3 / 2}
+ 6x + 5\sqrt x + 4} \right)^2}
\begin{cases}
 { > 0,} & {0 < x < 1} \\
 { < 0,} & {x > 1} \\
\end{cases}
\]
\noindent and
\[
\beta _{\Psi J\_\Psi K_0 }  = \mathop {\sup }\limits_{x \in (0,\infty )}
g_{\Psi J\_\Psi K_0 } (x) = \mathop {\lim }\limits_{x \to 1} g_{\Psi J\_\Psi
K_0 } (x) = \frac{4}{3}.
\]

\smallskip
\noindent
\textbf{(x) For} $D_{\Psi K_0 } (P\vert \vert Q) \leqslant \frac{1}{3}D_{F\Delta } (P\vert \vert Q)$\textbf{: }We have
\[
g_{\Psi K_0 \_F\Delta } (x)  = \frac{{f}''_{\Psi K_0 } (x)}{{f}''_{F\Delta }
(x)} = \frac{4\sqrt x \left( {4x^2 + 5x^{3 / 2} + 6x + 5\sqrt x + 4}
\right)\left( {x + 1} \right)^3}{\left( {\begin{array}{l}
 15 + 90x + 257x^2 + 492x^3 + 257x^4 + \\
 + 90x^5 + 15x^6 + 30\sqrt x + 150x^{3 / 2} + \\
 + 364x^{5 / 2} + 364x^{7 / 2} + 150x^{9 / 2} + 30x^{11 / 2} \\
 \end{array}} \right)},
 \]
 \begin{align}
{g}'_{\Psi K_0 \_F\Delta } (x) &=
 - \frac{12\left( {x + 1} \right)^2\left( {x - 1} \right)^3}{\sqrt x \left(
{\begin{array}{l}
 15 + 90x + 257x^2 + 492x^3 + 257x^4 + \\
 + 90x^5 + 15x^6 + 30\sqrt x + 150x^{3 / 2} + \\
 + 364x^{5 / 2} + 364x^{7 / 2} + 150x^{9 / 2} + 30x^{11 / 2} \\
 \end{array}} \right)^2}\times\notag\\
& \hspace{20pt} \times \left( {\begin{array}{l}
 10 + 110x + 486x^2 + 740x^3 + 486x^4 + \\
 + 110x^5 + 10x^6 + \,25\sqrt x + 205x^{3 / 2} + \\
 + 586x^{5 / 2} + 586x^{7 / 2} + 205x^{9 / 2} + 25x^{11 / 2} \\
 \end{array}} \right)
 \begin{cases}
 { > 0,} & {0 < x < 1} \\
 { < 0,} & {x > 1} \\
\end{cases}\notag
\end{align}
\noindent and
\[
\beta _{\Psi K_0 \_F\Delta }  = \mathop {\sup }\limits_{x \in (0,\infty )}
g_{\Psi K_0 \_F\Delta } (x) = \mathop {\lim }\limits_{x \to 1} g_{\Psi K_0
\_F\Delta } (x) = \frac{1}{3}.
\]

\smallskip
\noindent
\textbf{(xi) For} $D_{\Psi T} (P\vert \vert Q) \leqslant \frac{3}{8}D_{FI} (P\vert \vert Q)$\textbf{: }We have
\[
g_{\Psi T\_FI} (x) = \frac{{f}''_{\Psi T} (x)}{{f}''_{FI} (x)} =
\frac{16\sqrt x \left( {x^2 + x + 1} \right)\left( {\sqrt x + 1}
\right)^2}{\left( {\begin{array}{l}
 15x^4 + 30x^{7 / 2} + 60x^3 + 90x^{5 / 2} + \\
 + 122x^2 + 90x^{3 / 2} + 60x + 30\sqrt x + 15 \\
 \end{array}} \right)},
 \]
 \begin{align}
{g}'_{\Psi T\_FI} (x) & =
 - \frac{8\left( {x - 1} \right)}{\sqrt x \left( {\begin{array}{l}
 15x^4 + 30x^{7 / 2} + 60x^3 + 90x^{5 / 2} + \\
 + 122x^2 + 90x^{3 / 2} + 60x + 30\sqrt x + 15 \\
 \end{array}} \right)^2}\times \notag\\
& \hspace{20pt} \times \left( {\begin{array}{l}
 15x^6 + 26x^{11 / 2} + 40x^5 + 86x^{9 / 2} + 9x^4 + \\
 + 9x^2 + 86x^{3 / 2} + 40x + 26\sqrt x + 15 + \\
 + (x^2 - 1)^2\sqrt x \left( {34x + 65\sqrt x + 34} \right) \\
 \end{array}} \right)
\begin{cases}
 { > 0,} & {0 < x < 1} \\
 { < 0,} & {x > 1} \\
\end{cases} \notag
\end{align}
\noindent and
\[
\beta _{\Psi T\_FI}  = \mathop {\sup }\limits_{x \in (0,\infty )} g_{\Psi
T\_FI} (x) = \mathop {\lim }\limits_{x \to 1} g_{\Psi T\_FI} (x) =
\frac{3}{8}.
\]

\smallskip
\noindent
\textbf{(xii) For} $D_{F\Delta } (P\vert \vert Q) \leqslant \frac{9}{8}D_{FI} (P\vert \vert Q)$\textbf{: }We have
\[
g_{F\Delta \_FI} (x) = \frac{{f}''_{F\Delta } (x)}{{f}''_{FI} (x)} =
\frac{\left( {\begin{array}{l}
 15 + 90x + 257x^2 + 492x^3 + 257x^4 + \\
 + 90x^5 + 15x^6 + 30\sqrt x + 150x^{3 / 2} + \\
 + 364x^{5 / 2} + 364x^{7 / 2} + 150x^{9 / 2} + 30x^{11 / 2} \\
 \end{array}} \right)}{\left( {x + 1} \right)^2\left( {\begin{array}{l}
 15x^4 + 30x^{7 / 2} + 60x^3 + 90x^{5 / 2} + \\
 + 122x^2 + 90x^{3 / 2} + 60x + 30\sqrt x + 15 \\
 \end{array}} \right)},
 \]
 \begin{align}
{g}'_{F\Delta \_FI} (x) & = - \frac{32x^{3 / 2}\left( {x - 1} \right)}{\left(
{x + 1} \right)^3\left( {\begin{array}{l}
 15x^4 + 30x^{7 / 2} + 60x^3 + 90x^{5 / 2} + \\
 + 122x^2 + 90x^{3 / 2} + 60x + 30\sqrt x + 15 \\
 \end{array}} \right)^2}\times\notag\\
& \hspace{20pt} \times \left( {\begin{array}{l}
 75x^5 + 300x^{9 / 2} + 675x^4 + 1200x^{7 / 2} + \\
 + 1682x^3 + 1928x^{5 / 2} + 1682x^2 + \\
 + 1200x^{3 / 2} + 675x + 300\sqrt x + 75 \\
 \end{array}} \right)
 \begin{cases}
 { > 0,} & {0 < x < 1} \\
 { < 0,} & {x > 1} \\
\end{cases} \notag
\end{align}
\noindent and
\[
\beta _{F\Delta \_FI}  = \mathop {\sup }\limits_{x \in (0,\infty )}
g_{F\Delta \_FI} (x) = \mathop {\lim }\limits_{x \to 1} g_{F\Delta \_FI} (x)
= \frac{9}{8}.
\]

\smallskip
\noindent
\textbf{(xiii) For} $D_{FI} (P\vert \vert Q) \leqslant \frac{16}{15}D_{Fh} (P\vert \vert Q)$\textbf{: }We have
\[
g_{FI\_Fh} (x) = \frac{{f}''_{FI} (x)}{{f}''_{Fh} (x)} = \frac{\left(
{\begin{array}{l}
 15x^4 + 30x^{7 / 2} + 60x^3 + 90x^{5 / 2} + \\
 + 122x^2 + 90x^{3 / 2} + 60x + 30\sqrt x + 15 \\
 \end{array}} \right)}{15 \left({x + 1} \right)^3\left( {\sqrt x + 1}
\right)^2},
\]
\[
{g}'_{FI\_Fh} (x) =
 - \frac{32x\left( {\sqrt x - 1} \right)\left( {2x + 3\sqrt x + 2}
\right)}{ 15 \left( {x + 1} \right)^4\left( {\sqrt x + 1} \right)^3}
\begin{cases}
 { > 0,} & {0 < x < 1} \\
 { < 0,} & {x > 1} \\
\end{cases}
\]
\noindent and
\[
\beta _{FI\_Fh}  = \mathop {\sup }\limits_{x \in (0,\infty )} g_{FI\_Fh} (x)
= \mathop {\lim }\limits_{x \to 1} g_{FI\_Fh} (x) = \frac{16}{15}.
\]

\smallskip
\noindent
\textbf{(xiv) For} $D_{Fh} (P\vert \vert Q) \leqslant \frac{15}{14}D_{FJ} (P\vert \vert Q)$\textbf{: }We have
\[
g_{Fh\_FJ} (x) = \frac{{f}''_{Fh} (x)}{{f}''_{FJ} (x)} = \frac{15\left(
{\sqrt x + 1} \right)^2\left( {x + 1} \right)^2}{\left( {\begin{array}{l}
 15x^3 + 30x^{5 / 2} + 45x^2 + \\
 + 44x^{3 / 2} + 45x + 30\sqrt x + 15 \\
 \end{array}} \right)},
 \]
 \[
{g}'_{Fh\_FJ} (x)  =
 - \frac{120\left( {x - 1} \right)\left( {x + 1} \right)\sqrt x \left( {3x +
4\sqrt x + 3} \right)}{\left( {\begin{array}{l}
 15x^3 + 30x^{5 / 2} + 45x^2 + \\
 + 44x^{3 / 2} + 45x + 30\sqrt x + 15 \\
 \end{array}} \right)^2}
\begin{cases}
 { > 0,} & {0 < x < 1} \\
 { < 0,} & {x > 1} \\
\end{cases}
\]
\noindent and
\[
\beta _{Fh\_FJ}  = \mathop {\sup }\limits_{x \in (0,\infty )} g_{Fh\_FJ} (x)
= \mathop {\lim }\limits_{x \to 1} g_{Fh\_FJ} (x) = \frac{15}{14}.
\]

\smallskip
\noindent
\textbf{(xv) For} $D_{FJ} (P\vert \vert Q) \leqslant \frac{7}{6}D_{FT} (P\vert \vert Q)$\textbf{: }We have
\[
g_{FJ\_FT} (x) = \frac{{f}''_{FJ} (x)}{{f}''_{FT} (x)} = \frac{\left( {x +
1} \right)\left( {\begin{array}{l}
 15x^3 + 30x^{5 / 2} + 45x^2 + \\
 + \,44x^{3 / 2} + 45x + 30\sqrt x + 15 \\
 \end{array}} \right)}{\left( {\begin{array}{l}
 15x^4 + 30x^{7 / 2} + 60x^3 + 58x^{5 / 2} + \\
 + 58x^2 + 58x^{3 / 2} + 60x + 30\sqrt x + 15 \\
 \end{array}} \right)},
 \]
 \begin{align}
{g}'_{FJ\_FT} (x) & =
 - \frac{8\left( {x - 1} \right)\sqrt x }{\left(
{\begin{array}{l}
 15x^4 + 30x^{7 / 2} + 60x^3 + 58x^{5 / 2} + \\
 + 58x^2 + 58x^{3 / 2} + 60x + 30\sqrt x + 15 \\
 \end{array}} \right)^2}\times \notag\\
& \hspace{20pt} \times \left( {\begin{array}{l}
 45x^4 + 180x^{7 / 2} + 360x^3 + 540x^{5 / 2} + \\
 + \,598x^2 + 540x^{3 / 2} + 360x + 180\sqrt x + 45 \\
 \end{array}} \right)
 \begin{cases}
 { > 0,} & {0 < x < 1} \\
 { < 0,} & {x > 1} \\
\end{cases} \notag
\end{align}
\noindent and
\[
\beta _{FJ\_FT}  = \mathop {\sup }\limits_{x \in (0,\infty )} g_{FJ\_FT} (x)
= \mathop {\lim }\limits_{x \to 1} g_{FJ\_FT} (x) = \frac{7}{6}.
\]

\smallskip
\noindent
\textbf{(xvi) For} $D_{FJ} (P\vert \vert Q) \leqslant \frac{7}{6}D_{FK_0 } (P\vert \vert Q)$\textbf{: }We have
\[
g_{FJ\_FK_0 } (x) = \frac{{f}''_{FJ} (x)}{{f}''_{FK_0 } (x)} = \frac{\left(
{\begin{array}{l}
 15x^3 + 30x^{5 / 2} + 45x^2 + \\
 + \,44x^{3 / 2} + 45x + 30\sqrt x + 15 \\
 \end{array}} \right)}{3\left( {5x^2 + 6x + 5} \right)\left( {\sqrt x + 1}
\right)^2},
\]
\[
{g}'_{FJ\_FK_0 } (x) =
 - \frac{4\left( {\sqrt x - 1} \right)\left( {\begin{array}{l}
 15x^3 + 30x^{5 / 2} + 65x^2 + \\
 + 68x^{3 / 2} + \,65x + 30\sqrt x + 15 \\
 \end{array}} \right)}{3\left( {5x^2 + 6x + 5} \right)^2
 \left( {\sqrt x + 1} \right)^3}
 \begin{cases}
 { > 0,} & {0 < x < 1} \\
 { < 0,} & {x > 1} \\
\end{cases}
\]
\noindent and
\[
\beta _{FJ\_FK_0 }  = \mathop {\sup }\limits_{x \in (0,\infty )} g_{FJ\_FK_0
} (x) = \mathop {\lim }\limits_{x \to 1} g_{FJ\_FK_0 } (x) = \frac{7}{6}.
\]

\smallskip
\noindent
\textbf{(xvii) For} $D_{FT} (P\vert \vert Q) \leqslant 2\,D_{F\Psi } (P\vert \vert Q)$\textbf{: }We have
\[
g_{FT\_F\Psi } (x)  = \frac{{f}''_{FT} (x)}{{f}''_{F\Psi } (x)} =
\frac{\left( {\begin{array}{l}
 15x^4 + 30x^{7 / 2} + 60x^3 + 58x^{5 / 2} + \\
 + \,58x^2 + 58x^{3 / 2} + 60x + 30\sqrt x + 15 \\
 \end{array}} \right)}{\left( {x + 1} \right)\left( {\begin{array}{l}
 15x^3 + 14x^{5 / 2} + 13x^2 + \\
 + \,12x^{3 / 2} + 13x + 14\sqrt x + 15 \\
 \end{array}} \right)},
 \]
 \begin{align}
{g}'_{FT\_F\Psi } (x) & = - \frac{8\left( {x - 1} \right)}{\sqrt x \left( {x +
1} \right)^2\left( {\begin{array}{l}
 15x^3 + 14x^{5 / 2} + 13x^2 + \\
 + 12x^{3 / 2} + 13x + 14\sqrt x + 15 \\
 \end{array}} \right)^2}\times \notag\\
& \hspace{20pt} \times \left( {\begin{array}{l}
 15x^6 + 60x^{11 / 2} + 105x^5 + 184x^{9 / 2} + \\
 + 265x^4 + 380x^{7 / 2} + 382x^3 + 380x^{5 / 2} + \\
 + 265x^2 + 184x^{3 / 2} + 105x + 60\sqrt x + 15 \\
 \end{array}} \right)
 \begin{cases}
 { > 0,} & {0 < x < 1} \\
 { < 0,} & {x > 1} \\
\end{cases}\notag
\end{align}
\noindent and
\[
\beta _{FT\_F\Psi }  = \mathop {\sup }\limits_{x \in (0,\infty )}
g_{FT\_F\Psi } (x) = \mathop {\lim }\limits_{x \to 1} g_{FT\_F\Psi } (x) =
2.
\]

\smallskip
\noindent
\textbf{(xviii) For} $D_{FK_0 } (P\vert \vert Q) \leqslant 2\,D_{F\Psi } (P\vert \vert Q)$\textbf{: }We have
\[
g_{FK_0 \_F\Psi } (x) = \frac{{f}''_{FK_0 } (x)}{{f}''_{F\Psi } (x)} =
\frac{3\left( {\sqrt x + 1} \right)^2\left( {5x^2 + 6x + 5} \right)}{\left(
{\begin{array}{l}
 15x^3 + 14x^{5 / 2} + 13x^2 + \\
 + \,12x^{3 / 2} + 13x + 14\sqrt x + 15 \\
 \end{array}} \right)},
 \]
 \begin{align}
{g}'_{FK_0 \_F\Psi } (x) & = - \frac{12\left( {x - 1} \right)}{\sqrt x \left(
{\begin{array}{l}
 15x^3 + 14x^{5 / 2} + 13x^2 + \\
 + 12x^{3 / 2} + 13x + 14\sqrt x + 15 \\
 \end{array}} \right)^2}\times\notag\\
& \hspace{20pt} \times \left( {\begin{array}{l}
 10x^4 + 25x^{7 / 2} + 58x^3 + 87x^{5 / 2} + \\
 + \,120x^2 + 87x^{3 / 2} + 58x + 25\sqrt x + 10 \\
 \end{array}} \right)
 \begin{cases}
 { > 0,} & {0 < x < 1} \\
 { < 0,} & {x > 1} \\
\end{cases}\notag
\end{align}
\noindent and
\[
\beta _{FK_0 \_F\Psi }  = \mathop {\sup }\limits_{x \in (0,\infty )} g_{FK_0
\_F\Psi } (x) = \mathop {\lim }\limits_{x \to 1} g_{FK_0 \_F\Psi } (x) = 2.
\]

Combining the parts (i)-(xix) we complete the proof of the theorem.
\end{proof}

\subsection{Unified Inequalities}

In \cite{tan4}, the author studied the following inequalities based on the first part of expression (\ref{eq15}):
\begin{align}
D_{I\Delta } \leqslant & \frac{2}{3}D_{h\Delta } \leqslant
\frac{1}{2}D_{J\Delta } \leqslant \frac{1}{3}D_{T\Delta } \leqslant D_{TJ}
\leqslant \frac{2}{3}D_{Th} \leqslant\notag\\
\label{eq25}
 & \leqslant 2D_{Jh} \leqslant \frac{1}{6}D_{\Psi \Delta } \leqslant
\frac{1}{5}D_{\Psi I} \leqslant \frac{2}{9}D_{\Psi h} \leqslant
\frac{1}{4}D_{\Psi J} \leqslant \frac{1}{3}D_{\Psi T} .
\end{align}

\noindent and
\begin{equation}
\label{eq26}
\frac{2}{3}D_{h\Delta } \leqslant 2D_{hI} \leqslant D_{TJ} .
\end{equation}

Combining the above inequalities given in (\ref{eq21}), (\ref{eq25}) and (\ref{eq26}), we get following unified result:
\begin{align}
 & D_{I\Delta } \leqslant \frac{2}{3} D_{h\Delta } \leqslant
\frac{1}{2}D_{J\Delta } \leqslant \left( {\begin{array}{l}
 \frac{1}{3}D_{T\Delta } \\
 2D_{hI} \\
 \end{array}} \right) \leqslant \left( {\begin{array}{l}
 \frac{1}{3}D_{K_0 \Delta } \\
 D_{TJ} \leqslant \frac{2}{3}D_{Th} \leqslant 2D_{Jh} \\
 \end{array}} \right) \leqslant\notag\\
 & \leqslant \frac{1}{2}D_{K_0 I} \leqslant \frac{2}{3}D_{K_0 h} \leqslant
\left( {\begin{array}{l}
 D_{K_0 J} \\
 \frac{1}{6}D_{\Psi \Delta } \\
 \end{array}} \right) \leqslant \frac{1}{5}D_{\Psi I} \leqslant
\frac{2}{9}D_{\Psi h} \leqslant \frac{1}{4}D_{\Psi J} \leqslant\notag\\
\label{eq27}
& \leqslant \frac{1}{3}\left(
{\begin{array}{l}
 D_{\Psi T} \\
 D_{\Psi K_0 } \leqslant \frac{1}{9}D_{F\Delta } \\
 \end{array}} \right) \leqslant \frac{1}{8}D_{FI} \leqslant \frac{2}{15}D_{Fh} \leqslant \frac{1}{7}D_{FJ} \leqslant
\frac{1}{6}\left( {\begin{array}{l}
 D_{FT} \\
 D_{FK_0 } \\
 \end{array}} \right) \leqslant \frac{1}{3}D_{F\Psi }.
\end{align}

\section{Inequalities among New Divergence Measures}

From the inequalities appearing in (\ref{eq24}), we observe that the measures (\ref{eq6})-(\ref{eq11}) bear the following relation
\[
B_2  (P\vert \vert Q) \leqslant \frac{1}{4}B_4 (P\vert \vert Q) \leqslant
\frac{1}{2}B_3 (P\vert \vert Q) \leqslant \frac{1}{16}B_6 (P\vert \vert Q) \leqslant \frac{1}{2}B_5 (P\vert \vert Q) \leqslant \frac{1}{4}B_1 (P\vert \vert Q),
\]

\noindent i.e.,
\begin{align}
D_{h\Delta } & (P\vert \vert Q) \leqslant \frac{1}{2}D_{K_0 \Delta } (P\vert
\vert Q) \leqslant D_{K_0 h} (P\vert \vert Q) \leqslant \notag\\
\label{eq28}
&\leqslant \frac{1}{4}D_{\Psi \Delta } (P\vert \vert Q) \leqslant \frac{1}{2}D_{\Psi K_0 } (P\vert \vert Q) \leqslant \frac{1}{4}D_{FK_0 }
(P\vert \vert Q).
\end{align}

The expression (\ref{eq28}) admits 15 nonnegative differences. These are as follows:
\begin{align}
L_1 (P\vert \vert Q) & = L_2 (P\vert \vert Q) = \frac{1}{2}L_3 (P\vert \vert
Q) = \frac{1}{2}D_{K_0 \Delta } (P\vert \vert Q) - D_{h\Delta } (P\vert \vert Q) \notag\\
& = \frac{1}{16}\sum\limits_{i = 1}^n {\frac{\left( {\sqrt {p_i } - \sqrt
{q_i } } \right)^6}{\sqrt {p_i q_i } \left( {p_i + q_i } \right)},}\notag\\
L_4 (P\vert \vert Q) & = \frac{1}{4}D_{\Psi \Delta } (P\vert \vert Q) - D_{K_0
h} (P\vert \vert Q) = \frac{1}{64}\sum\limits_{i = 1}^n {\frac{\left( {\sqrt
{p_i } - \sqrt {q_i } } \right)^8}{p_i q_i \left( {p_i + q_i } \right)}},\notag\\
L_5 (P\vert \vert Q) & = \frac{1}{4}D_{\Psi \Delta } (P\vert \vert Q) -
\frac{1}{2}D_{K_0 \Delta } (P\vert \vert Q)  = \frac{1}{64}\sum\limits_{i =
1}^n {\frac{\left( {p_i - q_i } \right)^2\left( {\sqrt {p_i } - \sqrt {q_i }
} \right)^4}{p_i q_i \left( {p_i + q_i } \right)}},\notag\\
L_6 (P\vert \vert Q) & = \frac{1}{4}D_{\Psi \Delta } (P\vert \vert Q) -
D_{h\Delta } (P\vert \vert Q) \notag\\
& = \frac{1}{64}\sum\limits_{i = 1}^n
{\frac{\left( {p_i + 6\sqrt {p_i q_i } + q_i } \right)\left( {\sqrt {p_i } -
\sqrt {q_i } } \right)^2}{p_i q_i \left( {p_i + q_i } \right)}},\notag\\
L_7 (P\vert \vert Q) & = \frac{1}{2}D_{\Psi K_0 } (P\vert \vert Q) -
\frac{1}{4}D_{\Psi \Delta } (P\vert \vert Q)\notag\\
& = \frac{1}{64}\sum\limits_{i = 1}^n {\frac{\left[ {2\left( {p_i + q_i }
\right) + \left( {\sqrt {p_i } - \sqrt {q_i } } \right)^2} \right]\left(
{p_i - q_i } \right)^2\left( {\sqrt {p_i } - \sqrt {q_i } } \right)^2}{p_i
q_i \left( {p_i + q_i } \right)}},\notag
\end{align}
\begin{align}
L_8 (P\vert \vert Q) & = \frac{1}{2}D_{\Psi K_0 } (P\vert \vert Q) - D_{K_0 h}
(P\vert \vert Q) \notag\\
&= \frac{1}{16}\sum\limits_{i = 1}^n {\frac{\left( {p_i +
q_i } \right)\left( {\sqrt {p_i } - \sqrt {q_i } } \right)^4}{p_i q_i }},\notag\\
L_9 (P\vert \vert Q) & = \frac{1}{2}D_{\Psi K_0 } (P\vert \vert Q) -
\frac{1}{2}D_{K_0 \Delta } (P\vert \vert Q)\notag\\
& = \frac{1}{16}\sum\limits_{i = 1}^n {\frac{\left( {p_i - q_i }
\right)^2\left[ {\left( {\sqrt {p_i } - \sqrt {q_i } } \right)^2 + \sqrt
{p_i q_i } } \right]\left( {\sqrt {p_i } - \sqrt {q_i } } \right)^2}{p_i q_i
\left( {p_i + q_i } \right)}},\notag\\
L_{10} (P\vert \vert Q) & = \frac{1}{2}D_{\Psi K_0 } (P\vert \vert Q) -
D_{h\Delta } (P\vert \vert Q) \notag\\
& = \frac{1}{16}\sum\limits_{i = 1}^n {\frac{\left[ {\sqrt {p_i q_i } \left(
{p_i + q_i } \right) + \left( {\sqrt {p_i } - \sqrt {q_i } } \right)^2}
\right]\left( {\sqrt {p_i } - \sqrt {q_i } } \right)^4}{p_i q_i \left( {p_i
+ q_i } \right)}},\notag\\
L_{11} (P\vert \vert Q) &= \frac{1}{4}D_{FK_0 } (P\vert \vert Q) -
\frac{1}{2}D_{\Psi K_0 } (P\vert \vert Q) \notag\\
& = \frac{1}{32}\sum\limits_{i = 1}^n {\frac{\left( {p_i + q_i }
\right)\left( {p_i - q_i } \right)^2\left( {\sqrt {p_i } - \sqrt {q_i } }
\right)^2}{p_i^{3 / 2} q_i^{3 / 2} }},\notag\\
L_{12} (P\vert \vert Q) & = \frac{1}{4}D_{FK_0 } (P\vert \vert Q) -
\frac{1}{4}D_{\Psi \Delta } (P\vert \vert Q)\notag\\
& = \frac{1}{64}\sum\limits_{i = 1}^n {\frac{\left[ {p_i + q_i + \sqrt {p_i
q_i } + \left( {\sqrt {p_i } - \sqrt {q_i } } \right)^2} \right]\left( {p_i
- q_i } \right)^4}{p_i^{3 / 2} q_i^{3 / 2} \left( {p_i + q_i } \right)}},\notag\\
L_{13} (P\vert \vert Q) & = \frac{1}{4}D_{FK_0 } (P\vert \vert Q) - D_{K_0 h}
(P\vert \vert Q)\notag\\
& = \frac{1}{64}\sum\limits_{i = 1}^n {\frac{\left( {p_i + q_i }
\right)\left( {p_i + 4\sqrt {p_i q_i } + q_i } \right)\left( {\sqrt {p_i } -
\sqrt {q_i } } \right)^4}{p_i^{3 / 2} q_i^{3 / 2} \left( {p_i + q_i }
\right)}},\notag\\
L_{14} (P\vert \vert Q) & = \frac{1}{4}D_{FK_0 } (P\vert \vert Q) -
\frac{1}{2}D_{K_0 \Delta } (P\vert \vert Q)\notag\\
& = \frac{1}{32}\sum\limits_{i = 1}^n {\frac{\left[ {p_i^2 + q_i^2 + 2\sqrt
{p_i q_i } \left( {p_i + q_i } \right)} \right]\left( {\sqrt {p_i } + \sqrt
{q_i } } \right)^2\left( {\sqrt {p_i } - \sqrt {q_i } } \right)^4}{p_i^{3 /
2} q_i^{3 / 2} \left( {p_i + q_i } \right)}}\notag
\intertext{and}
L_{15} (P\vert \vert Q) & = \frac{1}{4}D_{FK_0 } (P\vert \vert Q) - D_{h\Delta
} (P\vert \vert Q)\notag\\
& = \frac{1}{32}\sum\limits_{i = 1}^n {\frac{\left( {\sqrt {p_i } + \sqrt
{q_i } } \right)^2\left( {\sqrt {p_i } - \sqrt {q_i } } \right)^4}{p_i^{3 /
2} q_i^{3 / 2} \left( {p_i + q_i } \right)}} \times\notag\\
& \hspace{20pt} \times \left[ {p_i^3 + q_i^3 + 4\sqrt {p_i q_i } \left( {p_i^2 + q_i^2 }
\right) + 7p_i q_i \left( {p_i + q_i } \right)} \right].\notag
\end{align}

\begin{theorem} The following inequalities hold:
\begin{align}
\label{eq29}
& L_1 (P\vert \vert Q) \leqslant \frac{1}{2}L_6 (P\vert \vert Q) \leqslant
D_{K_0 T} (P\vert \vert Q) \leqslant L_5 (P\vert \vert Q).\\
\intertext{and}
\label{eq30}
& L_7 (P\vert \vert Q) \leqslant \begin{cases}
 {L_{8} (P\vert \vert Q) \leqslant \frac{1}{3}\left(
{{\begin{array}{*{20}c}
 {L_{12} (P\vert \vert Q)} \\
 {L_{13} (P\vert \vert Q)} \\
\end{array} }} \right) \leqslant \frac{1}{2}L_{11} (P\vert \vert Q)} \\
 {\left( {{\begin{array}{*{20}c}
 {L_8 (P\vert \vert Q)} \\
 {L_9 (P\vert \vert Q)} \\
\end{array} }} \right) \leqslant \frac{1}{3}\begin{cases}
 {L_{14} (P\vert \vert Q) \leqslant \frac{1}{2}L_{11} (P\vert \vert Q)} \\
 {L_{15} (P\vert \vert Q)}\\
\end{cases}}.\\
\end{cases}.
\end{align}
\end{theorem}

\begin{proof} We shall prove the above theorem in parts. Following the similar lines of part (i) of Theorem 3.1, it is sufficient to write in each case, the expressions similar to (\ref{eq22})-(\ref{eq24}). The rest part of the proof follows by the application of Lemma 1.2.

\bigskip
\noindent
\textbf{(i) For }$L_1 (P\vert \vert Q) \leqslant \frac{1}{6}L_6 (P\vert \vert Q)$\textbf{: }We have
\[
g_{L_1 \_L_6 } (x) = \frac{{f}''_{L_1 } (x)}{{f}''_{L_6 } (x)} = \frac{\sqrt
x \left( {\begin{array}{l}
 3x^3 + 12x^{5 / 2} + 25x^2 + \\
 + 40x^{3 / 2} + 25x + 12\sqrt x + 3 \\
 \end{array}} \right)}{2\left( {\begin{array}{l}
 x^4 + 4x^{7 / 2} + 13x^3 + 24x^{5 / 2} + \\
 + 36x^2 + 24x^{3 / 2} + 13x + 4\sqrt x + 1 \\
 \end{array}} \right)},
 \]
 \begin{align}
{g}'_{L_1 \_L_6 } (x) & = - \frac{3\left( {x - 1} \right)\left( {x + 1}
\right)^2}{4\sqrt x \left( {\begin{array}{l}
 x^4 + 4x^{7 / 2} + 13x^3 + 24x^{5 / 2} + \\
 + 36x^2 + 24x^{3 / 2} + 13x + 4\sqrt x + 1 \\
 \end{array}} \right)^2}\times\notag\\
 &\hspace{20pt} \times \left( {\begin{array}{l}
 x^4 + 8x^{7 / 2} + 27x^3 + 64x^{5 / 2} + \\
 + 80x^2 + 64x^{3 / 2} + 27x + 8\sqrt x + 1 \\
 \end{array}} \right)
\begin{cases}
 { > 0,} & {0 < x < 1} \\
 { < 0,} & {x > 1} \\
\end{cases}\notag
\end{align}
\noindent and
\[
\beta _{L_1 \_L_6 }  = \mathop {\sup }\limits_{x \in (0,\infty )} g_{L_1
\_L_6 } (x) = \mathop {\lim }\limits_{x \to 1} g_{L_1 \_L_6 } (x) =
\frac{1}{2}.
\]

\smallskip
\noindent
\textbf{(ii) For }$L_1 (P\vert \vert Q) \leqslant \frac{3}{2}D_{K_0 T} (P\vert
\vert Q)$\textbf{: }We have
\[
g_{L_1 \_D_{K_0 T} } (x) = \frac{{f}''_{L_1 } (x)}{{f}''_{D_{K_0 T} } (x)} =
\frac{\left( {\begin{array}{l}
 3x^3 + 12x^{5 / 2} + 25x^2 + \\
 + 40x^{3 / 2} + 25x + 12\sqrt x + 3 \\
 \end{array}} \right)}{2\left( {x + 1} \right)^2\left( {3x + 4\sqrt x + 3}
\right)^2},
\]
\[
{g}'_{L_1 \_D_{K_0 T} } (x)  = - \frac{2(x - 1)\left( {\begin{array}{l}
 3x^3 + 12x^{5 / 2} + 35x^2 + \\
 + 40x^{3 / 2} + 35x + 12\sqrt x + 3 \\
 \end{array}} \right)}{2\sqrt x \left( {x + 1} \right)^3\left( {3x + 4\sqrt
x + 3} \right)^2}
\begin{cases}
 { > 0,} & {0 < x < 1} \\
 { < 0,} & {x > 1} \\
\end{cases}
\]
\noindent and
\[
\beta_{L_1 \_D_{K_0 T} }  = \mathop {\sup }\limits_{x \in (0,\infty )}
g_{L_1 \_D_{K_0 T} } (x) = \mathop {\lim }\limits_{x \to 1} g_{L_1 \_D_{K_0
T} } (x) = \frac{3}{2}.
\]

\smallskip
\noindent
\textbf{(iii) For }$K_0 T$\textbf{: }We have
\[
g_{D_{K_0 T} \_L_5 } (x) = \frac{{f}''_{K_0 T} (x)}{{f}''_{L_5 } (x)} =
\frac{2\sqrt x \left( {x + 1} \right)^2\left( {3x + 4\sqrt x + 3}
\right)}{\left( {\begin{array}{l}
 2x^4 + 5x^{7 / 2} + 14x^3 + 23x^{5 / 2} + \\
 + 32x^2 + 23x^{3 / 2} + 14x + 5\sqrt x + 2 \\
 \end{array}} \right)},
 \]
 \begin{align}
{g}'_{D_{K_0 T} \_L_5 } (x) & = - \frac{2\left( {x - 1} \right)^3\left( {x +
1} \right)\left( {x + \sqrt x + 1} \right)}{\sqrt x \left( {\begin{array}{l}
 2x^4 + 5x^{7 / 2} + 14x^3 + 23x^{5 / 2} + \\
 + 32x^2 + 23x^{3 / 2} + 14x + 5\sqrt x + 2 \\
 \end{array}} \right)^2}\times \notag\\
& \hspace{20pt} \times  \left({\begin{array}{l}
 3x^2 + 5x^{3 / 2} + 14x + 5\sqrt x + 3 \\
 \end{array}} \right)\begin{cases}
 {> 0,} & {0 < x < 1} \\
 {< 0,} & {x > 1} \\
\end{cases}\notag
\end{align}
\noindent and
\[
\beta _{D_{K_0 T} \_L_5 }  = \mathop {\sup }\limits_{x \in (0,\infty )}
g_{D_{K_0 T} \_L_5 } (x) = \mathop {\lim }\limits_{x \to 1} g_{D_{K_0 T}
\_L_5 } (x) = \frac{2}{3}.
\]

Combining the parts (i)-(iii), we get the proof of (\ref{eq28}).

\smallskip
\noindent
\textbf{(iv) For }$L_7 (P\vert \vert Q) \leqslant L_8 (P\vert \vert Q)$\textbf{: }We have
\[
g_{L_7 \_L_8 } (x) = \frac{{f}''_{L_7 } (x)}{{f}''_{L_8 } (x)} =
\frac{\left( {\begin{array}{l}
 3x^5 + 3x^{9 / 2} + 12x^4 + 10x^{7 / 2} + 17x^3 + \\
 + 6x^{5 / 2} + 17x^2 + 10x^{3 / 2} + 12x + 3\sqrt x + 3 \\
 \end{array}} \right)}{2\left( {x + 1} \right)^3\left( {2x^2 + x^{3 / 2} +
\sqrt x + 2} \right)},
\]
\begin{align}
{g}'_{L_7 \_L_8 } (x) & = - \frac{3\left( {x - 1} \right)\left( {\sqrt x - 1}
\right)^2}{4\sqrt x \left( {x + 1} \right)^4\left( {2x^2 + x^{3
/ 2} + \sqrt x + 2} \right)^2}\times \notag\\
& \hspace{30pt} \times\left( {\begin{array}{l}
 x^5 + 6x^{9 / 2} + 16x^4 + 34x^{7 / 2} + \\
 + 35x^3 + 40x^{5 / 2} + 35x^2 + \\
 + 34x^{3 / 2} + 16x + 6\sqrt x + 1 \\
 \end{array}} \right)
\begin{cases}
 { > 0,} & {0 < x < 1} \\
 { < 0,} & {x > 1} \\
\end{cases}\notag
\end{align}
\noindent and
\[
\beta _{L_7 \_L_8 } = \mathop {\sup }\limits_{x \in (0,\infty )} g_{L_7
\_L_8 } (x) = \mathop {\lim }\limits_{x \to 1} g_{L_7 \_L_8 } (x) = 1.
\]

\smallskip
\noindent
\textbf{(v) For }$L_7 (P\vert \vert Q) \leqslant L_9 (P\vert \vert Q)$\textbf{: }We have
\[
g_{L_7 \_L_9 } (x) = \frac{{f}''_{L_7 } (x)}{{f}''_{L_9 } (x)} =
\frac{2\left( {\begin{array}{l}
 3x^5 + 3x^{9 / 2} + 12x^4 + 10x^{7 / 2} + 17x^3 + \\
 + 6x^{5 / 2} + 17x^2 + 10x^{3 / 2} + 12x + 3\sqrt x + 3 \\
 \end{array}} \right)}{\left( {\begin{array}{l}
 8x^5 + 7x^{9 / 2} + 30x^4 + 20x^{7 / 2} + 31x^3 + 31x^2 + \\
 + 3x^2\left( {\sqrt x - 1} \right)^2 + 20x^{3 / 2} + 30x + 7\sqrt x + 8 \\
 \end{array}} \right)},
 \]
 \begin{align}
{g}'_{L_7 \_L_9 } (x) & = - \frac{3\left( {x - 1} \right)\left( {x + 1}
\right)^2}{\sqrt x \left( {\begin{array}{l}
 8x^5 + 7x^{9 / 2} + 30x^4 + 20x^{7 / 2} + 31x^3 + 31x^2 + \\
 + 3x^2\left( {\sqrt x - 1} \right)^2 + 20x^{3 / 2} + 30x + 7\sqrt x + 8 \\
 \end{array}} \right)^2}\times\notag\\
 & \hspace{20pt} \times \left( {\begin{array}{l}
 x^6 + 4x^{11 / 2} + 17x^5 + 48x^{9 / 2} + 131x^4 + \\
 + 172x^{7 / 2} + 214x^3 + 172x^{5 / 2} + \\
 + 131x^2 + 48x^{3 / 2} + 17x + 4\sqrt x + 1 \\
 \end{array}} \right)
\begin{cases}
 { > 0,} & {0 < x < 1} \\
 { < 0,} & {x > 1} \\
\end{cases}\notag
\end{align}
\noindent and
\[
\beta _{L_7 \_L_9 }  = \mathop {\sup }\limits_{x \in (0,\infty )} g_{L_7 L_9
} (x) = \mathop {\lim }\limits_{x \to 1} g_{L_7 L_9 } (x) = 1.
\]

\smallskip
\noindent
\textbf{(vi) For }$L_8 (P\vert \vert Q) \leqslant \frac{1}{3}L_{12} (P\vert \vert
Q)$\textbf{: }We have
\[
g_{L_8 \_L_{12} } (x) = \frac{{f}''_{L_8 } (x)}{{f}''_{L_{12} } (x)} =
\frac{8\sqrt x \left( {2x^2 + x^{3 / 2} + \sqrt x + 2} \right)\left( {x + 1}
\right)^3}{\left( {\sqrt x + 1} \right)^2\left( {\begin{array}{l}
 2\left( {\sqrt x - 1} \right)^2\left( {x^4 + 5x^3 + 12x^2 + 5x + 1} \right)
+ \\
 + \left( {x + 1} \right)\left( {13x^4 + 38x^3 + 42x^2 + 38x + 13} \right)
\\
 \end{array}} \right)},
 \]
 \begin{align}
{g}'_{L_8 \_L_{12} } (x) & = - \frac{24\left( {\sqrt x - 1} \right)\left( {x +
1} \right)^2}{\sqrt x \left( {\sqrt x + 1} \right)^3\left(
{\begin{array}{l}
 2\left( {\sqrt x - 1} \right)^2\left( {x^4 + 5x^3 + 12x^2 + 5x + 1} \right)
+ \\
 + \left( {x + 1} \right)\left( {13x^4 + 38x^3 + 42x^2 + 38x + 13} \right)
\\
 \end{array}} \right)^2}\notag\\
& \hspace{20pt} \times \left( {\begin{array}{l}
 5x^8 + 5x^{15 / 2} + 21x^7 + 19x^{13 / 2} + 52x^6 + \\
 \; + 49x^{11 / 2} + 155x^5 + 87x^{9 / 2} + 174x^4 + \\
 \; + 87x^{7 / 2} + 155x^3 + 49x^{5 / 2} + 52x^2 + \\
 \; + 19x^{3 / 2} + 21x + 5\sqrt x + 5 \\
 \end{array}} \right)
\begin{cases}
 { > 0,} & {0 < x < 1} \\
 { < 0,} & {x > 1} \\
\end{cases}\notag
\end{align}
\noindent and
\[
\beta _{L_8 \_L_{12} }  = \mathop {\sup }\limits_{x \in (0,\infty )} g_{L_8
\_L_{12} } (x) = \mathop {\lim }\limits_{x \to 1} g_{L_8 \_L_{12} } (x) =
\frac{1}{3}.
\]

\smallskip
\noindent
\textbf{(vii) For }$L_8 (P\vert \vert Q) \leqslant \frac{1}{3}L_{13} (P\vert \vert Q)$\textbf{: }We have
\[
g_{L_8 \_L_{13} } (x)  = \frac{{f}''_{L_8 } (x)}{{f}''_{L_{13} } (x)} =
\frac{8\sqrt x \left( {2x^2 + x^{3 / 2} + \sqrt x + 2} \right)}{3\left(
{\sqrt x + 1} \right)^2\left( {5x^2 + 2x + 5} \right)},
\]
\begin{align}
{g}'_{L_8 \_L_{13} } (x) & = - \frac{8\left( {\sqrt x - 1} \right)}{3\sqrt x \left( {\sqrt x + 1} \right)^3\left( {5x^2 +
2x + 5} \right)^2}\times \notag\\
& \hspace{20pt} \times \left(
{\begin{array}{l}
 5x^4 + 5x^{7 / 2} + 3x^3 + 7x^{5 / 2} + \\
 + 20x^2 + 7x^{3 / 2} + 3x + 5\sqrt x + 5 \\
 \end{array}} \right)
\begin{cases}
 { > 0,} & {0 < x < 1} \\
 { < 0,} & {x > 1} \\
\end{cases}\notag
\end{align}
\noindent and
\[
\beta _{L_8 \_L_{13} }  = \mathop {\sup }\limits_{x \in (0,\infty )} g_{L_8
\_L_{13} } (x) = \mathop {\lim }\limits_{x \to 1} g_{L_8 \_L_{13} } (x) =
\frac{1}{3}.
\]

\smallskip
\noindent
\textbf{(viii) For }$L_{12} (P\vert \vert Q) \leqslant \frac{3}{2}L_{11} (P\vert \vert Q)$\textbf{: }We have
\[
g_{L_{12} \_L_{11} } (x) = \frac{{f}''_{L_{12} } (x)}{{f}''_{L_{11} } (x)} =
\frac{\left( {\sqrt x + 1} \right)^2\left( {\begin{array}{l}
 2\left( {\sqrt x - 1} \right)^2\left( {x^4 + 5x^3 + 12x^2 + 5x + 1} \right)
+ \\
 + \left( {x + 1} \right)\left( {13x^4 + 38x^3 + 42x^2 + 38x + 13} \right)
\\
 \end{array}} \right)}{\left( {x + 1} \right)^3\left( {\begin{array}{l}
 15x^3 + 14x^{5 / 2} + 13x^2 + \\
 + 12x^{3 / 2} + 13x + 14\sqrt x + 15 \\
 \end{array}} \right)},
 \]
 \begin{align}
{g}'_{L_{12} \_L_{11} } (x) & = - \frac{6\left( {x - 1} \right)}{\sqrt x
\left( {x + 1} \right)^4\left( {\begin{array}{l}
 15x^3 + 14x^{5 / 2} + 13x^2 + \\
 + 12x^{3 / 2} + 13x + 14\sqrt x + 15 \\
 \end{array}} \right)^2}\times \notag\\
& \hspace{20pt} \times \left( {\begin{array}{l}
 15x^8 + 30x^{15 / 2} + 121x^7 + 174x^{13 / 2} + 346x^6 + \\
 + 178x^{11 / 2} + 487x^5 + 258x^{9 / 2} + 622x^4 + \\
 + 258x^{7 / 2} + 487x^3 + 178x^{5 / 2} + 346x^2 + \\
 + 174x^{3 / 2} + 121x + 30\sqrt x + 15 \\
 \end{array}} \right)\begin{cases}
 { > 0,} & {0 < x < 1} \\
 { < 0,} & {x > 1} \\
\end{cases} \notag
\end{align}
\noindent and
\[
\beta _{L_{12} \_L_{11} }  = \mathop {\sup }\limits_{x \in (0,\infty )}
g_{L_{12} \_L_{11} } (x) = \mathop {\lim }\limits_{x \to 1} g_{L_{12}
\_L_{11} } (x) = \frac{3}{2}.
\]

\smallskip
\noindent
\textbf{(ix) For }$L_{13} (P\vert \vert Q) \leqslant \frac{3}{2}L_{11} (P\vert \vert Q)$\textbf{: }We have
\[
g_{L_{13} \_L_{11} } (x)  = \frac{{f}''_{L_{13} } (x)}{{f}''_{L_{11} } (x)} =
\frac{3\left( {\sqrt x + 1} \right)^2\left( {5x^2 + 2x + 5} \right)}{\left(
{\begin{array}{l}
 15x^3 + 14x^{5 / 2} + 13x^2 + \\
 + 12x^{3 / 2} + 13x + 14\sqrt x + 15 \\
 \end{array}} \right)},
 \]
\[
{g}'_{L_{13} \_L_{11} } (x) = - \frac{24\left( {x - 1} \right)\left(
{\begin{array}{l}
 5x^4 + 5x^{7 / 2} + 3x^3 + 7x^{5 / 2} + \\
 + 20x^2 + 7x^{3 / 2} + 3x + 5\sqrt x + 5 \\
 \end{array}} \right)}{\sqrt x \left( {\begin{array}{l}
 15x^3 + 14x^{5 / 2} + 13x^2 + \\
 + 12x^{3 / 2} + 13x + 14\sqrt x + 15 \\
 \end{array}} \right)^2}
\begin{cases}
 { > 0,} & {0 < x < 1} \\
 { < 0,} & {x > 1} \\
\end{cases}
\]
\noindent and
\[
\beta _{L_{13} \_L_{11} }  = \mathop {\sup }\limits_{x \in (0,\infty )}
g_{L_{13} \_L_{11} } (x) = \mathop {\lim }\limits_{x \to 1} g_{L_{13}
\_L_{11} } (x) = \frac{3}{2}.
\]

\smallskip
\noindent
\textbf{(x) For }$L_8 (P\vert \vert Q) \leqslant \frac{1}{3}L_{14} (P\vert \vert Q)$\textbf{: }We have
\[
g_{L_8 \_L_{14} } (x)  = \frac{{f}''_{L_8 } (x)}{{f}''_{L_{14} } (x)} =
\frac{8\sqrt x \left( {2x^2 + x^{3 / 2} + \sqrt x + 2} \right)\left( {x + 1}
\right)^3}{\left( {\begin{array}{l}
 15x^6 + 30x^{11 / 2} + 72x^5 + 114x^{9 / 2} + \\
 + 137x^4 + 160x^{7 / 2} + 96x^3 + 160x^{5 / 2} + \\
 + 137x^2 + 114x^{3 / 2} + 72x + 30\sqrt x + 15 \\
 \end{array}} \right)},
 \]
 \begin{align}
{g}'_{L_8 \_L_{14} } (x) & = - \frac{24(x - 1)(x + 1)^2}{\sqrt x \left(
{\begin{array}{l}
 15x^6 + 30x^{11 / 2} + 72x^5 + 114x^{9 / 2} + \\
 + 137x^4 + 160x^{7 / 2} + 96x^3 + 160x^{5 / 2} + \\
 + 137x^2 + 114x^{3 / 2} + 72x + 30\sqrt x + 15 \\
 \end{array}} \right)^2}\times \notag\\
& \hspace{20pt} \times \left( {\begin{array}{l}
 5x^8 + 5x^{15 / 2} + 21x^7 + 19x^{13 / 2} + 60x^6 + \\
 + 73x^{11 / 2} + 343x^5 + 159x^{9 / 2} + 270x^4 + \\
 + 159x^{7 / 2} + 243x^3 + 73x^{5 / 2} + 60x^2 + \\
 + 19x^{3 / 2} + 21x + 5\sqrt x + 5 \\
 \end{array}} \right)\begin{cases}
 { > 0,} & {0 < x < 1} \\
 { < 0,} & {x > 1} \\
\end{cases}\notag
\end{align}
\noindent and
\[
\beta _{L_8 \_L_{14} } = \mathop {\sup }\limits_{x \in (0,\infty )} g_{L_8
\_L_{14} } (x) = \mathop {\lim }\limits_{x \to 1} g_{L_8 \_L_{14} } (x) =
\frac{1}{3}.
\]

\smallskip
\noindent
\textbf{(xi) For }$L_8 (P\vert \vert Q) \leqslant \frac{1}{3}L_{15} (P\vert \vert Q)$\textbf{: }We have
\[
g_{L_8 \_L_{15} } (x)  = \frac{{f}''_{L_8 } (x)}{{f}''_{L_{15} } (x)} =
\frac{8\sqrt x \left( {2x^2 + x^{3 / 2} + \sqrt x + 2} \right)\left( {x + 1}
\right)^3}{\left( {\begin{array}{l}
 15x^6 + 30x^{11 / 2} + 78x^5 + 126x^{9 / 2} + \\
 + 145x^4 + 164x^{7 / 2} + 36x^3 + 164x^{5 / 2} + \\
 + 145x^2 + 126x^{3 / 2} + 78x + 30\sqrt x + 15 \\
 \end{array}} \right)},
 \]
 \begin{align}
{g}'_{L_8 \_L_{15} } (x) & = - \frac{24(x - 1)(x + 1)^2}{\sqrt x \left(
{\begin{array}{l}
 15x^6 + 30x^{11 / 2} + 78x^5 + 126x^{9 / 2} + \\
 + 145x^4 + 164x^{7 / 2} + 36x^3 + 164x^{5 / 2} + \\
 + 145x^2 + 126x^{3 / 2} + 78x + 30\sqrt x + 15 \\
 \end{array}} \right)^2}\times \notag\\
& \hspace{20pt} \times \left( {\begin{array}{l}
 5x^8 + 5x^{15 / 2} + 19x^7 + 11x^{13 / 2} + 58x^6 + \\
 + 81x^{11 / 2} + 365x^5 + 223x^{9 / 2} + 386x^4 + \\
 + 223x^{7 / 2} + 365x^3 + 81x^{5 / 2} + 58x^2 + \\
 + 11x^{3 / 2} + 19x + 5\sqrt x + 5 \\
 \end{array}} \right)\begin{cases}
 { > 0,} & {0 < x < 1} \\
 { < 0,} & {x > 1} \\
\end{cases}\notag
\end{align}
\noindent and
\[
\beta _{L_8 \_L_{15} }  = \mathop {\sup }\limits_{x \in (0,\infty )} g_{L_8
\_L_{15} } (x) = \mathop {\lim }\limits_{x \to 1} g_{L_8 \_L_{15} } (x) =
\frac{1}{3}.
\]

\smallskip
\noindent
\textbf{(xii) For }$L_9 (P\vert \vert Q) \leqslant \frac{1}{3}L_{14} (P\vert \vert Q)$\textbf{: }We hav
\[
g_{L_9 \_L_{14} } (x)  = \frac{{f}''_{L_9 } (x)}{{f}''_{L_{14} } (x)} =
\frac{2\sqrt x \left( {\begin{array}{l}
 8x^5 + 7x^{9 / 2} + 30x^4 + 20x^{7 / 2} + 31x^3 + 31x^2 + \\
 + 3x^2\left( {\sqrt x - 1} \right)^2 + 20x^{3 / 2} + 30x + 7\sqrt x + 8 \\
 \end{array}} \right)}{\left( {\begin{array}{l}
 15x^6 + 30x^{11 / 2} + 72x^5 + 114x^{9 / 2} + \\
 + 137x^4 + 160x^{7 / 2} + 96x^3 + 160x^{5 / 2} + \\
 + 137x^2 + 114x^{3 / 2} + 72x + 30\sqrt x + 15 \\
 \end{array}} \right)},
 \]
 \begin{align}
{g}'_{L_9 \_L_{14} } (x) & = - \frac{6\left( {x - 1} \right)\left( {\sqrt x -
1} \right)^2\left( {x + 1} \right)^2}{\sqrt x \left( {\begin{array}{l}
 15x^6 + 30x^{11 / 2} + 72x^5 + 114x^{9 / 2} + \\
 + 137x^4 + 160x^{7 / 2} + 96x^3 + 160x^{5 / 2} + \\
 + 137x^2 + 114x^{3 / 2} + 72x + 30\sqrt x + 15 \\
 \end{array}} \right)^2}\times \notag\\
& \hspace{20pt} \times \left( {\begin{array}{l}
 20x^7 + 75x^{13 / 2} + 274x^6 + 634x^{11 / 2} + \\
 + 1274x^5 + 1685x^{9 / 2} + 2144x^4 + \\
 + 2124x^{7 / 2} + 2144x^3 + 1685\ast x^{5 / 2} + \\
 + 1274x^2 + 634x^{3 / 2} + 274x + 75\sqrt x + 20 \\
 \end{array}} \right)\left\{ {{\begin{array}{*{20}c}
 { > 0,} & {0 < x < 1} \\
 { < 0,} & {x > 1} \\
\end{array} }} \right.\notag
\end{align}
\noindent and
\[
\beta _{L_9 \_L_{14} }  = \mathop {\sup }\limits_{x \in (0,\infty )} g_{L_9
\_L_{14} } (x) = \mathop {\lim }\limits_{x \to 1} g_{L_9 \_L_{14} } (x) =
\frac{1}{3}.
\]

\smallskip
\noindent
\textbf{(xiii) For }$L_9 (P\vert \vert Q) \leqslant \frac{1}{3}L_{15} (P\vert \vert
Q)$\textbf{: }We have
\[
g_{L_9 \_L_{15} } (x) = \frac{{f}''_{L_9 } (x)}{{f}''_{L_{15} } (x)} =
\frac{2\sqrt x \left( {\begin{array}{l}
 8x^5 + 7x^{9 / 2} + 30x^4 + 20x^{7 / 2} + 31x^3 + 31x^2 + \\
 + 3x^2\left( {\sqrt x - 1} \right)^2 + 20x^{3 / 2} + 30x + 7\sqrt x + 8 \\
 \end{array}} \right)}{\left( {\begin{array}{l}
 15x^6 + 30x^{11 / 2} + 78x^5 + 126x^{9 / 2} + \\
 + 145x^4 + 164x^{7 / 2} + 36x^3 + 164x^{5 / 2} + \\
 + 145x^2 + 126x^{3 / 2} + 78x + 30\sqrt x + 15 \\
 \end{array}} \right)},
 \]
 \begin{align}
{g}'_{L_9 \_L_{15} } (x) & = - \frac{6\left( {x - 1} \right)\left( {\sqrt x -
1} \right)^2\left( {x + 1} \right)^2}{\sqrt x \left( {\begin{array}{l}
 15x^6 + 30x^{11 / 2} + 78x^5 + 126x^{9 / 2} + \\
 + 145x^4 + 164x^{7 / 2} + 36x^3 + 164x^{5 / 2} + \\
 + 145x^2 + 126x^{3 / 2} + 78x + 30\sqrt x + 15 \\
 \end{array}} \right)^2}\times \notag\\
& \hspace{20pt} \times \left( {\begin{array}{l}
 20x^8 + 35x^{15 / 2} + 136x^7 + 129x^{13 / 2} + \\
 + 272x^6 - 197x^{11 / 2} + 536x^5 - 223x^{9 / 2} + \\
 + 504x^4 - 223x^{7 / 2} + 536x^3 - 197x^{5 / 2} + \\
 + 272x^2 + 129x^{3 / 2} + 136x + 35\sqrt x + 20 \\
 \end{array}} \right)\begin{cases}
 { > 0,} & {0 < x < 1} \\
 { < 0,} & {x > 1} \\
\end{cases}\notag
\end{align}
\noindent and
\[
\beta _{L_9 \_L_{15} }  = \mathop {\sup }\limits_{x \in (0,\infty )} g_{L_9
\_L_{15} } (x) = \mathop {\lim }\limits_{x \to 1} g_{L_9 \_L_{15} } (x) =
\frac{1}{3}.
\]

\smallskip
\noindent
\textbf{(xiv) For }$L_{14} (P\vert \vert Q) \leqslant \frac{3}{2}L_{11} (P\vert \vert Q)$\textbf{: }We have
\[
g_{L_{14} \_L_{11} } (x)  = \frac{{f}''_{L_{14} } (x)}{{f}''_{L_{11} } (x)} =
\frac{2\sqrt x \left( {\begin{array}{l}
 15x^6 + 30x^{11 / 2} + 72x^5 + 114x^{9 / 2} + \\
 + 137x^4 + 160x^{7 / 2} + 96x^3 + 160x^{5 / 2} + \\
 + 137x^2 + 114x^{3 / 2} + 72x + 30\sqrt x + 15 \\
 \end{array}} \right)}{\left( {x + 1} \right)^3\left( {\begin{array}{l}
 15x^3 + 14x^{5 / 2} + 13x^2 + \\
 + 12x^{3 / 2} + 13x + 14\sqrt x + 15 \\
 \end{array}} \right)},
 \]
 \begin{align}
{g}'_{L_{14} \_L_{11} } (x) & = - \frac{6\left( {x - 1} \right)\left( {\sqrt x
- 1} \right)^2}{\sqrt x \left( {x + 1} \right)^4\left( {\begin{array}{l}
 15x^3 + 14x^{5 / 2} + 13x^2 + \\
 + 12x^{3 / 2} + 13x + 14\sqrt x + 15 \\
 \end{array}} \right)^2}\times \notag\\
& \hspace{20pt} \times \left( {\begin{array}{l}
 20x^7 + 75x^{13 / 2} + 274x^6 + 634x^{11 / 2} + \\
 + 1274x^5 + 1685x^{9 / 2} + 2144x^4 + \\
 + 2124x^{7 / 2} + 2144x^3 + 1685x^{5 / 2} + \\
 + 1274x^2 + 634x^{3 / 2} + 274x + 75\sqrt x + 20 \\
 \end{array}} \right)\begin{cases}
 { > 0,} & {0 < x < 1} \\
 { < 0,} & {x > 1} \\
\end{cases}\notag
\end{align}
\noindent and
\[
\beta _{L_{14} \_L_{11} }  = \mathop {\sup }\limits_{x \in (0,\infty )}
g_{L_{14} \_L_{11} } (x) = \mathop {\lim }\limits_{x \to 1} g_{L_{14}
\_L_{11} } (x) = \frac{3}{2}.
\]

Combining the parts (iii)-(xiv) we get the proof of (b). Finally, the parts (i)-(xiv) completes the proof of the theorem.
\end{proof}

\subsection{Relationships with the Terms of Exponential Divergence Series}

In this section we shall relate the first four terms of the series (\ref{eq17}).

\begin{theorem} The following inequalities hold
\begin{equation}
\label{eq31}
L_5 (P\vert \vert Q) \le \frac{1}{2048}K_2 (P\vert \vert Q)
\end{equation}

\noindent and
\begin{equation}
\label{eq32}
L_4 (P\vert \vert Q) \le \frac{1}{32768}K_3 (P\vert \vert Q)
\end{equation}
\end{theorem}

\begin{proof} We shall use the same arguments as of Theorem 3.1 to prove this theorem.

\bigskip
\noindent
\textbf{(i) For }$L_5 (P\vert \vert Q) \le \frac{1}{2048}K_2 (P\vert \vert Q)$\textbf{: }We have
\[
g_{L_5 \_D_{K_2 } } (x)  = \frac{{f}''_{L_5 } (x)}{{f}''_{D_{K_2 } } (x)} =
\frac{x^{3 / 2}\left( {\begin{array}{l}
 2x^4 + 5x^{7 / 2} + 14x^3 + 23x^{5 / 2} + \\
 + 32x^2 + 23x^{3 / 2} + 14x + 5\sqrt x + 2 \\
 \end{array}} \right)}{80\left( {x + 1} \right)^3\left( {\sqrt x + 1}
\right)^4\left( {7x^2 + 10x + 7} \right)},
\]
\begin{align}
{g}'_{L_5 \_D_{K_2 } } (x) & = - \frac{3\sqrt x \left( {\sqrt x - 1}
\right)}{80\left( {\sqrt x + 1} \right)^5\left( {x + 1} \right)^4\left(
{7x^2 + 10x + 7} \right)^2}\times \notag\\
& \hspace{20pt} \times \left( {\begin{array}{l}
 7x^7 + 28x^{13 / 2} + 106x^6 + 262x^{11 / 2} + \\
 + 560x^5 + 860x^{9 / 2} + 1151x^4 + \\
 + 1220x^{7 / 2} + 1151x^3 + 860x^{5 / 2} + \\
 + 560x^2 + 262x^{3 / 2} + 106x28\sqrt x + 7 \\
 \end{array}} \right)\begin{cases}
 { > 0,} & {0 < x < 1} \\
 { < 0,} & {x > 1} \\
\end{cases}\notag
\end{align}
\noindent and
\[
\beta _{L_5 \_D_{K_2 } }  = \mathop {\sup }\limits_{x \in (0,\infty )} g_{L_5
\_D_{K_2 } } (x) = \mathop {\lim }\limits_{x \to 1} g_{L_5 \_D_{K_2 } } (x)
= \frac{1}{2048}.
\]

\noindent
\textbf{(ii) For }$L_4 (P\vert \vert Q) \le \frac{1}{32768}K_3 (P\vert \vert Q)$\textbf{: }We have
\[
g_{L_4 \_D_{K_3 } } (x)  = \frac{{f}''_{L_4 } (x)}{{f}''_{D_{K_3 } } (x)} =
\frac{x^{5 / 2}\left( {x^2 + x^{3 / 2} + 3x + \sqrt x + 1} \right)}{56\left(
{\sqrt x + 1} \right)^4\left( {x + 1} \right)^3(9x^2 + 14x + 9)},
\]
\begin{align}
{g}'_{L_4 \_D_{K_3 } } (x) & = - \frac{3x^{3 / 2}\left( {\sqrt x - 1}
\right)}{112\left( {\sqrt x + 1} \right)^5\left( {x + 1} \right)^4\left(
{9x^2 + 14x + 9} \right)^2}\times\notag\\
& \hspace{20pt} \times \left( {\begin{array}{l}
 15x^5 + 36x^{9 / 2} + 116x^4 + 166x^{7 / 2} \\
 + 261x^3 + 252x^{5 / 2} + 261x^2 + \\
 + 166x^{3 / 2} + 116x + 36\sqrt x + 15 \\
 \end{array}} \right)\begin{cases}
 { > 0,} & {0 < x < 1} \\
 { < 0,} & {x > 1} \\
\end{cases}\notag
\end{align}
\noindent and
\[
\beta _{L_4 \_D_{K_3 } } = \mathop {\sup }\limits_{x \in (0,\infty )} g_{L_4
\_D_{K_3 } } (x) = \mathop {\lim }\limits_{x \to 1} g_{L_4 \_D_{K_3 } } (x)
= \frac{1}{32768}.
\]
\end{proof}

\begin{remark} After simplifications, we have following relations with the first four terms of the exponential divergence series:
\begin{equation}
\label{eq33}
\frac{1}{2}F(P\vert \vert Q)= K_0 (P\vert \vert Q) + \frac{1}{4}K_1 (P\vert \vert Q),
\end{equation}
\begin{align}
\label{eq34}
\frac{1}{4}\Psi (P\vert \vert Q) + \Delta (P\vert \vert Q) & \leqslant K_0
(P\vert \vert Q) + \frac{1}{128}K_2 (P\vert \vert Q) \notag\\
& \leqslant 8T(P\vert \vert Q) + \frac{3}{256}K_2 (P\vert \vert Q)
\end{align}
\noindent and
\begin{align}
\frac{1}{2}\Psi (P\vert \vert Q) + 32h(P\vert \vert Q)  & \leqslant 2\Delta
(P\vert \vert Q) + 4K_0 (P\vert \vert Q) + \frac{1}{1024}K_3 (P\vert \vert
Q)\notag\\
\label{eq35}
& \leqslant 5K_0 (P\vert \vert Q) + \frac{1}{1024}K_3 (P\vert \vert Q).
\end{align}

The expression (\ref{eq33}) relates the measures $K_0 (P\vert \vert Q)$ and $K_1 (P\vert \vert Q)$. The expression (\ref{eq34}) relates $K_0 (P\vert \vert Q)$ and $K_2 (P\vert \vert Q)$ and the expression (\ref{eq35}) relates $K_0 (P\vert \vert Q)$ and $K_3 (P\vert \vert Q)$ with the other known measures.
\end{remark}

\end{document}